\documentclass[a4paper,11pt]{article}
\title{Polyhedral Clinching Auctions for Two-sided Markets}
\author{Hiroshi Hirai\thanks{The University of Tokyo, Tokyo, Japan.
 hirai@mist.i.u-tokyo.ac.jp}\ \ and\ \ Ryosuke Sato\thanks{The University of Tokyo, Tokyo, Japan. ryosuke.sato.517@gmail.com}}

\date{\today}
\usepackage{amsmath}
\usepackage{algorithm}
\usepackage{algorithmic}
\usepackage{amsthm}
\usepackage{comment}
\usepackage[dvipdfmx]{graphicx}
\usepackage{amsfonts}
\usepackage{color}
\usepackage[top=1truein,bottom=1truein,left=1truein,right=1truein]{geometry}

  \newtheorem{theorem}{Theorem}[section]
  \newtheorem{proposition}[theorem]{Proposition}
  \newtheorem{corollary}[theorem]{Corollary}
  
  \newtheorem{lemma}[theorem]{Lemma}
  \newtheorem*{claim}{Claim}
  \newtheorem{remark}{Remark}
  \newtheorem{assumption}[theorem]{Assumption}
\begin{document}
\maketitle

\begin{abstract}
		In this paper, 
		we present a new model and two mechanisms for auctions 
		in two-sided markets of buyers and sellers, 
		where budget constraints are imposed on buyers.
		Our model incorporates polymatroidal environments, 
		and is applicable to a wide variety of models 
		that include multiunit auctions, matching markets 
		and reservation exchange markets.
		Our mechanisms are build on polymatroidal network 
		flow model by Lawler and Martel, and enjoy various 
		nice properties such as incentive compatibility of 
		buyers, individual rationality, pareto optimality, 
		strong budget balance.
		The first mechanism is a simple 
		``reduce-to-recover" algorithm
		that reduces the market to be one-sided,
		applies the polyhedral clinching auction by Goel et al,
		and lifts the resulting allocation to
		the original two-sided market via polymatroidal 
		network flow.
		The second mechanism is a two-sided generalization 
		of the polyhedral clinching auction,
		which improves the first mechanism 
		in terms of the fairness of revenue sharing on sellers.
		Both mechanisms are implemented 
		by polymatroid algorithms.
		We demonstrate how our framework is 
		applied to internet display ad auctions.
		\end{abstract}

\section{Introduction}
	Mechanism design for auctions in {\it two-sided markets}
	is a challenging and urgent issue, 
	especially for rapidly growing
	fields of internet advertisement.
	In ad-exchange platforms, the owners of 
	websites want to get revenue
	by selling their ad slots, 
	and the advertisers want to purchase ad slots.
	Auctions are an efficient way of mediating them, 
	allocating ad slots, and determining payments and revenues, 
	where the underlying market is two-sided in principle.
	Similar situations arise from stock exchanges and 
	spectrum license reallocation; see e.g., 
	\cite{BKLT2016,DRT2014}.
	Despite its potential applications,
	auction theory for two-sided markets is currently 
	far from dealing with such real-world markets.
	The main difficulty is that the auctioneer 
	has to consider incentives of buyers and sellers, 
	both possibly strategic, 
	and is confronted with
	impossibility theorems to design a mechanism  
	achieving both accuracy and efficiency,
	even in the simplest case of bilateral trade~\cite{M1983}.

	In this paper, we address auctions for two-sided markets, 
	aiming to overcome such difficulties and
	provide a reasonable and implementable framework.
	To capture realistic models mentioned above, 
	we deal with {\it budget constraints} on buyers.
	The presence of budgets drastically changes the situation 
	in which traditional auction theory  is not applicable.
	Our investigation is thus based on two recent seminal works on
	auction theory of budgeted one-sided markets:
	\begin{itemize}
	\item[(i)] Dobzinski et al. \cite{DLN2012} presented the first 
	effective framework for budget-constrained markets. 
	Generalizing the celebrated clinching
	framework by Ausubel \cite{A2004}, they
	proposed an incentive compatible, 
	individually rational, and pareto
	optimal mechanism, called the 
	``{\it Adaptive Clinching Auction}", for markets in
	which the budget information is public to the auctioneer.
	This work triggered subsequent works 
	dealing with more complicated
	settings \cite{BCMX2010,BHLS2015,DHS2015,FLSS2011,
	GMP2013,GMP2014,GMP2015}.
	\item[(ii)] Goel et al. \cite{GMP2015} utilized 
	polymatroid theory to generalize the above result
	for a broader class of auction models including
	previously studied budgeted settings as well as new models for 
	contemporary auctions such as Adwords Auctions.
	Here a polymatroid is a polytope 
	associated with a monotone submodular
	function, and can represent 
	the space of feasible transactions under several natural 
	constraints.
	They presented a polymatroid-oriented clinching mechanism, called 
	``{\it Polyhedral Clinching Auction}," for markets with polymatroidal 
	environments.
	This mechanism enjoys incentive compatibility, individual 
	rationality,
	and pareto optimality, and can be implemented 
	via efficient submodular
	optimization algorithms that have been developed in the literature of
	combinatorial optimization \cite{F2005,S2003}.
	\end{itemize}

The goal of this paper is to extend this line of research to
reasonable two-sided settings.

	\paragraph{Our contribution.}	
	
	 We present a new model and mechanisms for auctions in
	two-sided markets. 
	Our market is modeled as a 
	bipartite graph of buyers and sellers, with transacting 
	goods through the links. The goods are 
	divisible and common in value.
	Each buyer wants the goods under a limited budget.
	Each seller constrains transactions of 
	his goods by a monotone
	submodular function on the 
	set of edges linked to him.
	Namely, possible transactions are restricted 
	to the corresponding polymatroid.
	In the auction, each buyer reports his bid 
	and budget to the auctioneer,
	and each seller reports his
	reserved price. 
	In our model, the reserved price is assumed to be identical with
	his true valuation; this assumption is crucial for avoiding impossibility theorems. 
	The utilities are quasi-linear 
	(within budget)
	on their valuations and payments/revenues. 
	The goal of this auction is to determine 
	transactions of goods, payments of buyers, and revenues
	of sellers, with which all participants are satisfied.
	In the case of a single seller, this model 
	coincides with that of Goel et al \cite{GMP2015}.
	
	For this model, we present 
	two mechanisms that satisfy
	the incentive compatibility of buyers, 
	individual rationality, pareto optimality 
	and strong budget balance.
	Our mechanisms are built on and analyzed via {\it polymatroidal 
	network flow} model by Lawler and Martel \cite{LM1982}. 
	This is a notable 
	feature of our technical contribution.
	It is the first to apply polymatroidal network flow to 
	mechanism design.

	The first mechanism is a ``reduce-and-recover'' algorithm 
	via a one-sided market: 
	The mechanism constructs 
	``the reduced one-sided market" by
	aggregating all sellers to one seller, applies 
	the original clinching auction of Goel et al. \cite{GMP2015} 
	to determine a transaction vector, payments of buyers, and 
	the total revenue of the seller. 
	The transaction vector of the original two-sided market
	is recovered by computing a polymatroidal network flow.
	The total revenue is distributed to the  
	original sellers arbitrarily 
	so that incentive rationality of sellers and strong budget 
	balance are satisfied. 
	We prove in Theorem \ref{mechanism_1} that this mechanism 
	satisfies the desirable properties mentioned above.
	Also this mechanism is implementable by 
	polymatroid algorithms.
	The characteristic of this mechanism is 
	to determine statically transactions and revenues of sellers 
	at the end. 
	This however can cause unfair revenue sharing, 
	which we will discuss in Section~\ref{sec:discussion}.

	The second mechanism is a two-sided generalization of 
	the polyhedral clinching auction by Goel et al. \cite{GMP2015},
	which determines dynamically
	transactions and revenues, and 
	improves fairness on the revenue sharing. 
	The mechanism works as the original clinching auction: 
	As price clocks increase, 
	each buyer clinches a maximal amount of goods
	not affecting other buyers.
	Here each buyer transacts with multiple sellers, 
	and hence conducts a multidimensional clinching.
	We prove in Theorem \ref{greedy} 
	an intriguing property that feasible 
	transactions of sellers for the clinching forms 
	a polymatroid, 
	which we call the {\em clinching polytope}, and moreover 
	the corresponding submodular function can be computed 
	in polynomial time. 
	Thus this mechanism is also implementable.
	We reveal in Theorem \ref{relation} that the allocation
	to the buyers obtained by this mechanism 
	is the same as that by
	the original polyhedral clinching auction 
	applied to the reduced one-side market.
	This means that the second mechanism achieves 
	the same performance for buyers 
	as that in the original one, and also 
	improves the first mechanism in terms of the revenues sharing on sellers.

	Our framework captures a wide 
	variety of auction models in
	two-sided markets, thanks to the strong 
	expressive power of polymatroids.
	Examples include two-sided extensions of multiunit auctions 
	\cite{DLN2012} and 
	matching markets \cite{FLSS2011} (for divisible goods), 
	and a version of reservation exchange
	markets \cite{GLMNP2016}.
	We demonstrate 
	how our framework is applied to
	auctions for display advertisements.
	In addition, our model can incorporate with 
	concave budget constraints in Goel et al. \cite{GMP2014}.
	Also our result can naturally extend to concave budget settings (Remark~\ref{rem:concave}).
	Thus our framework is applicable to more complex settings occurring
	in the real world auctions, such as average budget constraints.

	\paragraph{Related work.}
	{\it Double auction} is the simplest auction for 
	two-sided markets, 
	where buyers and sellers 
	have unit demand and unit supply, respectively. 
	The famous {\it Myerson-Satterthwaite} 
	impossibility theorem \cite{M1983}
	says that there is no mechanism which simultaneously  
	satisfies incentive compatibility (IC), 
	individual rationality (IR), pareto optimality (PO),
	and budget balance (BB). 
	McAfee \cite{MA1992} proposed a mechanism 
	that satisfies (IC),(IR), and~(BB).
	Recently, Colini-Baldeschi et al. \cite{BKLT2016} proposed a
	mechanism that satisfies (IC), (IR), 
	and strong budget balance (SBB), and achieves an 
	$O(1)$-approximation to the maximum social welfare. 
	
	Goel et al. \cite{GLMNP2016} considered 
	a two-sided market model,
	called a {\it reservation exchange market}, 
	for internet advertisement.
		They formulated several axioms of mechanisms for this model,
	and presented an (implementable) mechanism satisfying 
	(IC) for buyers, 
	(IR), maximum social welfare, and 
	a fairness concept for sellers,
	called {\it $\alpha$-envy-freeness}.
	This mechanism is also based on the clinching framework,
	and sacrifices (IC) for sellers to avoid the impossibility
	theorem.
	Their setting is non-budgeted.
	
	Freeman et al. \cite{FPW2017} formulated
	the problem of wagering as an auction 
	in a special two-sided market, and 
	presented a mechanism, called the ``Double Clinching Auctions",
	satisfying (IC), (IR), and~(BB).
	They verified by computer simulations 
	that the mechanism shows
	near-pareto optimality.
	This mechanism is regarded as the first generalization of the 
	clinching framework to budgeted two-sided settings,
	though it is specialized to wagering.
	
	Our results in this paper provide the first 
	generic framework for auctions in budgeted two-sided markets.	
	
	\paragraph{Organization of this paper.}
	The rest of this paper is organized as follows.
	In Section~\ref{sec:result}, we introduce our model and 
	present the main result and applications.
	In Section~\ref{sec:mechanisms}, we present and analyze our mechanisms.
	In Section~\ref{sec:discussion}, we discuss our mechanisms and raise future research issues.
	In Section~\ref{sec:proof}, we give proofs.
	
	\paragraph{Notation.}
	
	Let $\mathbb R_{+}$ denote the set of 
	nonnegative real numbers, and let
	$\mathbb R_{+}^{E}$ denote the set of all 
	functions from a set $E$ to $\mathbb R_{+}$.
	For $w\in\mathbb R_+^E$, 
	we often denote $w(e)$ by $w_e$, 
	and write as $w = (w_e)_{e \in E}$.
	Also we denote $F\cup \{e\}$ by $F+e$, and 
	denote $F\setminus \{e\}$ by $F-e$.
	For $F \subseteq E$, let $w|_{F}$ 
	denote the restriction of $w$ 
	to $F$.
	Also let $w(F)$ denote the sum of $w(e)$ over $e \in F$, i.e., $w(F) := \sum_{e \in F}w (e)$.

	Let us recall theory of polymatroids and 
	submodular functions; see \cite{F2005,S2003}.
	A {\it monotone submodular function} on set $E$ is a function $f: 2^E \to \mathbb R_+$ satisfying:
	\begin{align}
	&f(\emptyset)=0, \nonumber \\
	&f(S) \leq f(T) \quad (S,T\subseteq E,\ S\subseteq T), \nonumber \\
	&f(S+e)-f(S) \geq f(T+e)-f(T)
	 \quad (S, T \subseteq E, S \subseteq T, e \in E\setminus T), \label{submo0}
	\end{align}
	where the third inequality is equivalent to the submodularity inequality:
	\begin{equation*}
	f(S)+f(T) \geq f(S\cap T)+f(S\cup T)	\quad (S,T\subseteq E).
	\end{equation*}
	The {\it polymatroid} $P:=P(f)$ associated with
	monotone submodular function $f:2^E \to \mathbb{R}_+$ 
	is defined by
	\begin{equation*}
	 P:=\bigl\{x\in\mathbb R^{E}_{+}\,|\, x(F) \leq f(F) \quad (F\subseteq E)\bigr\},
	\end{equation*}
	and the {\it base polytope} of $f$ is defined by 
	\begin{equation*}
	B:=\bigl\{x\in P\,|\, x(E)=f(E)\bigr\},
	\end{equation*}
	which is equal to the set of all maximal points in $P$.
	A point in $B$ is obtained by 
	the greedy algorithm 
	in polynomial time, provided the value of 
	$f$ for each subset $F\subseteq E$ 
	can be computed in polynomial time.

\section{Main result}\label{sec:result}
	We consider a two-sided market consisting of 
	$n$ buyers and $m$ sellers. 
	Our market is modeled as a bipartite graph $(N,M,E)$
	of disjoint sets $N$, $M$ of nodes and edge set 
	$E\subseteq N\times M$, where $N$ and $M$ represent the sets of
	buyers and sellers, respectively, 
	and buyer $i\in N$ and seller $j\in M$ are adjacent
	if and only if $i$ wants the goods of seller $j$.
	An edge $(i,j)\in E$ is denoted by $ij$.

	For buyer $i$ (resp. seller $j$),
	let $E_i$ (resp. $E_j$) denote the set of edges incident to $i$ (resp. $j$).
	In the market, the goods are divisible and homogeneous.
	Each buyer $i$ has three nonnegative 
	real numbers $v_i, v_i', B_i\in \mathbb R_+$,
	where $v_i$ and $v_i'$ are his {\it valuation} and {\it bid}, 
	respectively, for one unit of the goods, and $B_i$ is his {\it budget}.
	Each buyer $i$ acts strategically 
	for maximizing his {\it utility} $u_i$
	(defined later), and hence his bid $v'_i$ 
	is not necessarily equal to the true valuation $v_i$.
	In this market, each buyer $i$ reports $v_i'$ and $B_i$ to 
	the auctioneer. 
	Each seller~$j$ also has a valuation $\rho_j$  
	$\in \mathbb R_+$ for one unit of the goods, and reports
	$\rho_j$ to 
	the auctioneer as the {\it reserved price} of his goods,
	the lowest price that he admits for the goods.
	In particular, he is assumed to be truthful (to avoid the 
	impossibility theorem, as in \cite{GLMNP2016}).
	He also has a monotone submodular function $f_j$ on $E_j$, which
	controls transactions of goods through $E_j$.
	The value $f_j(F)$ for $F \subseteq E_j$ means 
	the maximum possible amount of goods transacted 
	through edge subset~$F$. In particular, $f_j(E_j)$ is 
	interpreted as his stock of goods.
	These assumptions on sellers are 
	characteristic of our model.

	Under this setting, the goal is to design a mechanism determining
	a reasonable allocation. An {\it allocation} $\mathcal A$ of the auction is a triple $\mathcal A:=(w, p, r)$
	of a {\it transaction vector}
	$w=(w_{ij})_{ij\in E}$, a {\it payment vector} 
	$p=(p_{i})_{i\in N}$,
	and a {\it revenue vector} $r=(r_{j})_{j\in M}$, 
	where $w_{ij}$ is the amount of transactions of goods 
	between buyer $i$ and seller $j$,
	$p_i$ is the payment of buyer $i$, and $r_j$ is the revenue of seller $j$.
	For each $j\in M$, 
	the restriction $w|_{E_j}$ of the transaction vector $w$ to $E_j\subseteq E$ must belong to
	the polymatroid $P_j$ corresponding to $f_j$:
	\begin{equation*}
	w|_{E_j}\in P_j \quad(j\in M).
	\end{equation*}
	Also the payment $p_i$ of buyer $i$ 
	must be within his budget $B_i$: 
	\begin{equation*}
	p_{i}\leq B_i\quad(i\in N).
	\end{equation*}

	A {\it mechanism} $\mathcal M$ is a function 
	that gives an allocation  $\mathcal A=(w,p,r)$
	 from {\it public information} $\mathcal I:=((N,M,E)$, 
	$\{v'_i\}_{i\in N}$,\ $\{B_i\}_{i\in N}$,\ 
	$\{\rho_j\}_{j\in M}$, 
	$\{f_{j}\}_{j\in M})$
	that the auctioneer can access.
	The true valuation $v_i$ of buyer 
	$i$ is private information that 
	only $i$ can access. We regard $\mathcal I$ 
	and $\{v_i\}_{i\in N}$ as the input of our model.

	Next we define the  {\it utilities} of buyers and sellers.
	For an allocation $\mathcal A=(w,p,r)$, 
	the utility $u_i(\mathcal A)$ of buyer $i$ is defined by:  
	\begin{equation}\label{eqn:utility_buyer}
	u_i(\mathcal A):=
	\begin{cases}
	\displaystyle 
	v_i w(E_i)-p_i\quad  {\rm if}\ p_i\leq B_i, \\
	-\infty \qquad\ \qquad \quad\  {\rm otherwise}.
	\end{cases}
	\end{equation}
	Namely the utility of a buyer 
	is the valuation of obtained goods minus the payment. 
	The utility of seller $j$ is defined by:
	\begin{equation} 
	\label{utility_seller}
	u_j(\mathcal A):=r_j+\rho_j (f_j (E_j)- w(E_j)).
	\end{equation}
	This is the sum of revenues and the 
	total valuation of his remaining goods.
	In this model,
	we consider the following properties of mechanism $\mathcal M$.  
	\begin{itemize}
	\item[(ICb)] {\it Incentive Compatibility} of buyers:  For every input
	$\mathcal I, \{v_i\}_{i\in N}$, it holds
	\begin{equation*}
	u_i(\mathcal M(\mathcal I))\leq u_i(\mathcal M(\mathcal I_i))\quad (i\in N),
	\end{equation*}
	 where $\mathcal I_i$ is obtained from $\mathcal I$ 
	by replacing bid $v'_i$ of buyer $i$ with his true valuation $v_i$.
	This means that it is the  
	best strategy for each buyer to 
	report his true valuation.
	\item[(IRb)] {\it Individual Rationality} of buyers: 
	For each buyer $i$, there is a bid $v'_i$ such that $i$ always 
	obtains nonnegative utility. 
	If (ICb) holds, then (IRb) is written as
	\begin{equation*}
	u_i (\mathcal M(\mathcal I_i))\geq 0 \quad(i\in N).
	\end{equation*}
	\item[(IRs)]{\it Individual Rationality} of sellers:
	The utility of each seller $j$ after the auction 
	is at least 
	the utility $\rho_j f_j(E_j)$ at the beginning:
	\begin{equation*}
	u_j (\mathcal M(\mathcal I))
	\geq \rho_j f_j(E_j) \quad(j\in M).
	\end{equation*}
	By (\ref{utility_seller}), (IRs) is equivalent to
	\begin{equation}
	\label{equivalent_to_IRs}
	r_j \geq \rho_j w(E_j).
	\end{equation}

	\item[(SBB)] {\it Strong Budget Balance}: 
	All payments of buyers 
	are directly given to sellers:
	\begin{equation*}
	\sum_{i\in N} p_i=\sum_{j\in M} r_j.
	\end{equation*}
	
	\item[(PO)] {\it Pareto Optimality}: There is no allocation $\mathcal A:=(w,p,r)$ which 
	satisfies $\displaystyle \sum_{i\in N}p_i\geq\sum_{j\in M} r_j$ and the 
	following three conditions:
	\begin{align*}
	&u_i(\mathcal M(\mathcal I^{\ast}))\leq u_i(\mathcal A)\quad\,(i\in N), \\
	&u_j(\mathcal M(\mathcal I^{\ast}))\leq u_j(\mathcal A)\quad (j\in M),
	\end{align*}
	and at least one of the inequalities holds 
	strictly, where $\mathcal I^{\ast}$ is obtained from 
	$\mathcal I$ 
	by replacing $\{v'_i\}_{i\in N}$ 
	with $\{v_i\}_{i\in N}$.
	Namely, there is no 
	other allocation superior to that given by 
	$\mathcal M$ for all buyers and sellers, 
	provided all buyers report their true valuation.
	\end{itemize}
	They are desirable properties that mechanisms should have.
	The main result is:
	\begin{theorem}
	\label{main}
	There exists a mechanism that satisfies all of (ICb),(IRb),(IRs),
	(SBB), and~(PO).
	\end{theorem}
	The details of our mechanisms are explained in Section 2.

	\begin{remark}
	\label{socialwelfare}
	{\rm Maximizing the {\it social welfare}, 
	the sum of utilities of all participants, 
	is usually set as the goal in traditional auction theory. 
	However, in budgeted settings, 
	it is shown in~\cite{DLN2012} that the
	maximum social welfare and incentive compatibility 
	of buyers  cannot be achieved simultaneously.
	As in the previous works 
	\cite{BCMX2010,BHLS2015,
	DLN2012,DHS2015,FLSS2011,GMP2013,GMP2014,GMP2015}, 
	we give priority 
	to incentive compatibility.}
	\end{remark}

	\subsection{Application to display ad auction of multiple sellers and slots}\label{subsec:application}
	We present applications of our results to 
	display advertisements\,(ads) between advertisers and 
	owners of websites. Each owner $j$ wants to sell  
		ad slots in his website. Each advertiser 
		$i$ wants to purchase the slots. Namely, the owners are sellers, and advertisers are buyers, where buyer $i$ is linked to seller $j$ if $i$ is interested in the slots of the website of $j$. The market is modeled as a bipartite graph $(N,M,E)$ as above. 
	We consider the following two types of ad auction,
	which are viewed as  
	{\it reservation exchange markets} 
	in the sense of \cite{GLMNP2016}.
	
	\subsubsection{Page-based ad auction}\label{subsub:page-based}
	The website of seller $j$ consists of pages 
	$1,2,\ldots,s_j$, where each page 
	$k\in \{1,2,\ldots,s_j\}$ has $t_j^k$ ad slots. 
	Each buyer purchases at most one slot from each page 
	(so that the same advertisement cannot be displayed 
	simultaneously).
	To control transaction of each seller $j$, set 
	$f_j:2^{E_j}\to \mathbb R_{+}$ as 
	\begin{equation}
	\label{appli}
	f_{j}(F_j):=\sum_{k\in\{1,2,\ldots,s_j\}} \min (t^k_j,|F_j|)\quad (F_j\subseteq E_j).
	\end{equation}
	Then $f_j$ is actually a monotone submodular function.
   It turns out that $f_j$ is appropriate for our purpose. 
   Hence the market falls into our model, and our 
   mechanisms are applicable to obtain transaction $w$.
   
   We first explain the way of 
   ad-slotting at owner $j$ from the obtained $w$, 
   in which the meaning of $f_j$ will be made clear.
   We consider the following network ${\cal N}_j$.
	Let $N_j$ denote the set of buyers linked to $j$.
	Consider buyers in $N_j$ and pages of $j$ as nodes in ${\cal N}_j$. 
	Add a directed edge from each buyer $i\in N_j$ to each page 
	$k$ with unit capacity.
	Add source node $a$ and edge $ai$ for each $i\in N_j$ 
	with capacity $w_{ij}$.
	Also, add sink node $b$ and edge $kb$ for each page $k$ with 
	capacity $t^{k}_{j}$.
	Consider a maximum flow $\varphi$ in the 
	resulting network $\mathcal N_j$.
	From the max-flow min-cut theorem, one can see that 
	$f_j(F_j)$ is nothing but the maximum value of a flow 
	from $F_j$ to $b$.
	Hence, transaction $w_{ij}$ 
	satisfies the polymatroid constraint of $j$ 
	if and only if 
	every maximum flow in $\mathcal N_j$ 
	attains capacity bound $w_{ij}$ on each source edge $ai$.
	
	From a maximum flow $\varphi$,
	the owner $j$ conducts ad-slotting according to values  
	$\varphi(ik)\in [0,1]$.
	Consider first 
	an ideal situation where $\varphi$ is integer-valued, 
	i.e., $\varphi(ik)\in \{0,1\}$.
	Then the owner~$j$ naturally assigns the ad of buyer $i$ 
	at page $k$ if $\varphi(ik)=1$, 
	since $w_{ij}$ is the sum of $\varphi(ik)$ over $k$ 
	(by flow conservation law and $\varphi(si)=w_{ij}$), 
	and at most $t_j^{k}$ buyers purchase page $k$.
	Consider the usual case where $\varphi(ik)$ 
	is not integer-valued.
	The value $\varphi(ik)$ will be interpreted as the probability 
	that the ad of $i$ is displayed at page $k$. 
	Add dummy buyers $i'$ to $N_j$ and define $\varphi(ki')$ 
	so that $|N_j| \geq t_j^k = \sum_{i \in N_j} \varphi(ik)$ (if necessary). 
	Consider probability $p = p^k$ over the set of 
	all $t_{j}^{k}$-element subsets of $N_j$ satisfying 
	\begin{equation}
	\label{probability}
	\sum_{X:i\in X}p(X)=\varphi(ik)\quad (i\in N_j).
	\end{equation}
	The owner selects a $t_{j}^k$-element subset $X \subseteq N_j$ with probability 
	$p(X)$, and displays the ads of $X$ at page $k$.
	By~(\ref{probability}), the ad of advertiser $i$ is displayed with 
	probability $\varphi(ik)$, as desired.
	
	The probability with 
	(\ref{probability}) can be constructed by 
	the following algorithm, where $t:=t_j^{k}$.
	\begin{itemize}
	\item[0.] Let $q_i:=\varphi(ik)\ (i\in N_j)$, and let
	$p(X):=0$ for all $t$-element subsets $X$ of $N_j$.
	\item[1.] If $q_i=0\ (i\in N_j)$, then output $p$.
	\item[2.] Sort $q_{i_1}\geq q_{i_2}\geq q_{i_3}\geq \cdots$, and  
	let $X:=\{i_1,\ldots,i_t\}$.
	\item[3.] Define $p(X)$ as the maximum $\gamma$ 
			with $0\leq q_i-\gamma$ for $i\in X$ and 
			$q_i\leq 1-\sum_{X}p(X)-\gamma$ for $i\notin X$.
	\item[4.] Let $q_i:=q_i-p(X)$ for $i\in X$, and go to 1.
	\end{itemize}
	Let us sketch the correctness of the algorithm. 
	In step 2, it always holds $0\leq q_i\leq 1-\sum_{X}p(X)$, and 
	$\sum_{i\in N_j}q_i=(1-\sum_{X}p(X))t$.
	Then it holds in step 3 that $q_i>0$ for $i\in X$, and 
	$q_i<1-\sum_{X}p(X)$ for $i\notin X$.
	Thus $\gamma$ is positive. After steps 3 and 4, 
	$q_i$ is zero for some $i\in X$  
	or $q_{i'}=1-\sum_{X}p(X)$ holds for some $i'\notin X$.
	For the latter case, $i'$ is always contained 
	in $X$ for the subsequent iterations. 
	If $q_i=1-\sum_{X}p(X)$ for all $i\in X$, then 
	$q_{i'}=0$ for $i'\notin X$.
	Thus the algorithm terminates after $O(N_j+t)$ iterations.
	By construction, it holds 
	$\sum_{X: i\in X}p(X)=\varphi (ik)\ (i\in N_j)$ and $\sum_{X} p(X) = 1$.

	\subsubsection{Quality-based ad auction}
	Each seller has one page with $t_j$ ad slots.
	Each slot $l\in \{1,2,\ldots,t_j\}$ has a barometer 
	$\beta^j_l\in \mathbb R_+$ for the quality of ads,
	which can be thought of as the expected number of views of $l$ 
	over a certain period. Suppose that 
	$\beta^j_1\geq \beta^j_2\geq \cdots \geq\beta^j_{t_j}$.
	The {\it view-impression} is the unit of this barometer;
	namely, slot $l$ has $\beta^j_l$ view-impressions.
	In the market, buyers purchase 
	view-impressions from sellers, 
	and have bids and valuations for the unit view-impression.
		After the auction, suppose that buyer 
	$i$ obtains $w_{ij}$ units of view-impressions from $j$. 
	Seller $j$ assigns ads of buyers 
	(linked to $j$) to his slots so that the same ads cannot be displayed in distinct slots at the same time. 
	 An ad-slotting is naturally represented by $(S,\psi)$ for a set  $S$ of buyers and a bijection 
	 $\psi$ from $S$ to the $|S|$ slots with the highest quality. 
	 Here the number of slots is assumed at least $|S|$ 
	 by adding dummy slots of $0$ view-impressions.
	 Seller $j$ displays ads of $i$ according to a probability 
	 distribution $p$ on the set of all ad-slottings satisfying 
	\begin{equation}
	\label{allocation}
	w_{ij}=\sum_{(S,\psi)} p(S,\psi)\beta^j_{\psi(i)}.
	\end{equation}
	Namely the expected number of 
	view-impressions of ads $i$ in the website of $j$
	is equal to $w_{ij}$.
	The existence of such a $p$ constrains transactions 
	$w_{ij}$, and is equivalent to the condition that  for each set $S$ 
	of buyers linked to $j$,
	$\sum_{i\in S}w_{ij}$ is at most the sum of $|S|$ 
	highest $\beta^j_l$. 
	This condition can be written by the following monotone 
	submodular function $f_j$:
	\begin{equation*}
	f_j(F_j):=\sum_{l=1}^{|F_j|}\beta^j_l \quad (F_j\subseteq E_j).
	\end{equation*}
    Then seller $j$ restricts his transactions $w_{ij}$ by 
	the corresponding polymatroid $P_j$, 
	to conduct the above way of ad-slotting after the auction. 
	Again the market falls into our model, and is a modification of 
	Adwords auctions in \cite{GMP2015} for display ad auctions 
	with multiple websites,
	where the way of ad-slotting 
	according to (\ref{allocation}) is 
	based on their idea.
	
	Notice that an extreme point of $P_j\subseteq \mathbb R^{E_j}$ 
	is precisely a vector $\xi=(\xi_i)_{i\in N_j}$ such that 
	for some subset $S\subseteq N_j$ and 
	bijection $\psi:S\to \{1,2,\ldots,|S|\}$, 
	it holds $\xi_i=\beta^j_{\psi(i)}$ if $i\in S$, or $0$ 
	otherwise.
	In particular, (\ref{allocation}) is viewed as a 
	convex combination of extreme points of~$P_j$.
	Therefore, the required probability distribution $p$ 
	 is obtained by expressing $w$ 
	as a convex combination of extreme points of~$P_j$.
	
	The page-based ad auction in Section~\ref{subsub:page-based} can incorporate   
	the quality-based formulation.
	Suppose that each slot 
	$l\in \{1,2,\ldots,t_j\}$ (including dummy slots) 
	on the page $k$ has a 
	barometer $\beta^{j}_{kl}$ for the quality of ads, 
	and that $\beta^j_{k1}\geq \beta^j_{k2}\geq \cdots\geq  \beta^j_{k t_j}$ for each 
	$k\in \{1,2,\ldots,t_j\}$. 
	Replace $f_j$ in (\ref{appli}) by 
		\begin{align}
	\label{appli2}
	f_j^{k}(F_j)&:=\sum_{j=1}^{|F_j|}\beta^j_{kl}, \\
	\label{appli3}
	f_{j}(F_j)&:=\sum_{k\in\{1,\ldots,s_j\}} 
	f_j^k(F_j)\quad (F_j\subseteq E_j).
	\end{align}
	This $f_j$ is also monotone submodular, and 
	our mechanisms are applicable.
	The obtained 
	transaction $w_{ij}$ is distributed to 
	view-impression $\varphi(ik)$ per page $k$ 
	so that $\sum_{k\in \{1,\ldots, t_j\}}\varphi(ik)=w_{ij}$, and 
	$\varphi(ik)\ (i\in N_j)$
	satisfies polymatroid constraint 
	by $f_j^k$. 
	Such a distribution can easily be obtained via polymatroidal 
	network flow, introduced in the next section.
	The owner conducts, at each page, 
	the ad-slotting in the same way as above.

\section{Polyhedral Clinching Auctions for Two-sided Markets}\label{sec:mechanisms}
\subsection{Polymatroidal network flow}\label{subsec:flow}
	In our mechanisms, we utilize polymatroidal network flow model 
	by Lawler and Martel \cite{LM1982}.
	A {\it polymatroidal network} is a directed network
	$(V,E)$ with source $s$ and sink $t$ such that 
	each node~$v$ has polymatroids $P^{+}_v$ and $P^{-}_v$
	defined on the sets $\delta_v^+$ and $\delta_v^-$ of
	edges leaving $v$ and entering~$v$, respectively.
	
	A {\it flow} is a function $\varphi:E\to\mathbb R_+$ satisfying
	\begin{align*}
	\varphi(\delta^+ v)&=\varphi(\delta^- v) 
	\quad (v\in V\setminus\{s,t\}),	\\
	\varphi|_{\delta^+ v}&\in P_v^+,\\
	\varphi|_{\delta^- v}&\in P_v^-.
	\end{align*}
	Let $f^+_v$ and $f^-_v$ denote the 
	monotone submodular functions corresponding 
	to $P^+_v$ and $P^-_v$, respectively. 
	In the case where the 
	network has edge-capacity 
	$c:E\to \mathbb R_+$, 
	the capacity-constraint $\varphi(e) \leq c(e)$ around a vertex $v$
	is also  written by
	the polymatroid of 
	submodular function:
	\begin{equation}\label{eqn:capacity}
	F \mapsto c(F) =  \sum_{e\in F}c(e) \quad 
	(F\subseteq \delta_v^+\ ({\rm or}\ \delta_v^-)).	
	\end{equation}
	The flow-value of a flow $\varphi$ is defined as
	\begin{equation*}
	\varphi(\delta^+_s)-
	\varphi(\delta^-_s) \ 
	(=\varphi( \delta^-_t)-
	\varphi(\delta^+_t)).
	\end{equation*}

	The following is a generalization of the 
	max-flow min-cut theorem 
	for polymatroidal networks, and is also 
	a version of the polymatroid 
	intersection theorem.
	\begin{theorem}[Lawler and Martel \cite{LM1982}]
	\label{max-flow-min-cut}
	The maximum value of a flow is equal to 
	\begin{equation}
	\label{max-flow-min-cut2}
		\min_{U,A,B} \{
		\sum_{v\in V\setminus U}f^-_
		v(\delta^-_v\cap A)+
		\sum_{v\in U}f^+_v(\delta^+_v\cap B)\},
	\end{equation}
	where $U$ ranges over all node subsets with 
	$s\in U\not\ni t$ and $\{A,B\}$ ranges over 
	all bi-partitions of the set of edges leaving $U$.
	\end{theorem}
	There are several algorithms to obtain a maximum flow $\varphi$ 
   and $U,A,B$ attaining the minimum of (\ref{max-flow-min-cut2}); 
   see \cite[Section 5]{F2005} and references therein.

\subsection{First Mechanism}
	Here we describe the first mechanism
	for Theorem \ref{main}.
	First we make the 
	following preprocessing on the market.
	Buyers and sellers are numbered as 
	$N = \{1,2,..., n\}$ and $M=\{1,2,...,m\}$. 
	For each seller $j \in M$, add to $N$ a {\em virtual buyer} $n+j$
	corresponding to $j$, and add to $E$ 
	a new edge connecting $n+j$ and $j$, i.e.,
	buyer $n+j$ transacts only with $j$.
	The virtual buyer $n+j$ has 
	sufficiently large budget $B_{n+j} = \infty$,
	and reports the valuation 
	$\rho_j$ of $j$ as bid $v'_{n+j}$.  
	The valuation  $v_{n+j}$ is set as 
	$v_{n+j}:=v'_{n+j}$ (though we do not use $v_{n+j}$ in the mechanisms).
	The submodular function $f_j$ is extended by
	\begin{align}
	\label{f_j}
	f_j (F):=
	\begin{cases}
	f_j (E_j)          \qquad  {\rm if}\ (n+j)j\in F, \\
	f_j (F)       \qquad\  {\rm otherwise}
	\end{cases}
	\quad (F \subseteq E_j \cup \{ (n+j)j\}).
	\end{align}
	This means that  the goods purchased by buyer $n+j$
	is interpreted as the unsold goods of $j$.
	The utility of seller $j$ (in the original market) 
	is the sum of the revenue of seller $j$ and the 
	utility of buyer $n+j$ after the auction.

	We utilize the framework by Goel et al. \cite{GMP2015} for
	one-sided markets under polymatroidal environments.
	From our two-sided market (including virtual buyers), 
	we construct a one-sided market consisting of 
	$N$ and one seller (= auctioneer) to which their framework is 
	applicable. Define polymatroid $P \subseteq \mathbb{R}_+^E$ by 
	\begin{equation}\label{eqn:P} 
	\displaystyle P:=\bigoplus_{j\in M}\ P_j
	=\{w\in \mathbb R_{+}^{E}\mid  
	 w|_{E_j}\in P_j\ (j\in M) \}.
	\end{equation}
	The polymatroidal environment on $N$ is defined by 
	the following polytope $\tilde{P}\subseteq \mathbb R^{N}_+$:
	\begin{equation*}
	\tilde{P}:
	=\{y\in \mathbb R_{+}^{N}\mid 
	 \exists w\in P,\  
	 y_i=w(E_i) \ (i\in N) \}.
	\end{equation*}
	Then $\tilde{P}$ is a polymatroid,
	which immediately follows from a network 
	induction of a polymatroid~\cite{MD1975}.
	The resulting one-sided market is called the {\em reduced one-side market}.
	An allocation of this market is a pair $(y,\pi)$ of 
	$y,\pi\in \mathbb R_+^{N}$,
	where $y_i$ and $\pi_i$ is the transaction and payment,
	respectively, of buyer $i$ to the seller.
	The transaction vector $y$ must belong to 
	the polymatroid~$\tilde{P}$.
	For each buyer $i$, payment $\pi_i$ must be 
	within his budget $B_i$.
	The utility of buyer $i$ is defined in (\ref{eqn:utility_buyer}) by replacing 
	$w(E_i)$ with $y_i$, and $p_i$ 
	with $\pi_i$.
	Now the original polyhedral clinching auction~\cite{GMP2015}
	is applicable to this one-sided market, and is given in 
	Algorithm \ref{algo1}.
	\begin{algorithm}[t]
	\caption{Polyhedral\ Clinching\ Auction \cite{GMP2015} for (Reduced) 
	One-sided Market.}         
	\label{algo1}                     
	\begin{algorithmic}[1]
	  \STATE $y_i:=0,\pi_i:=0,\ 
	  c_i:=0,\ d_i:=\infty\ 
	  (i\in N)\ {\rm and}\ l:=1$\\
	  \WHILE{$d\neq 0$}
	  \STATE Clinch a maximal increase 
	  $(\zeta_{i})_{i\in N}$   
	   \underline{{\it not affecting other buyers.}} 
	  \STATE  $y_i:=y_i+\zeta_i$,\ 
	  $\pi_i:=\pi_i+c_i \zeta_i\quad (i\in N)$ 
	  \STATE $c_{l}:=c_{l}+\varepsilon$ 
	  \STATE $ d_i:=
	 \begin{cases}
	(B_i-\pi_i)/c_i\quad {\rm if}\  c_i<v'_i,\\
	0 \quad\qquad\qquad\ \  {\rm otherwise}, 
	\end{cases}\quad (i\in N).$
	  \STATE $l:=l+1\ {\rm mod\ } m+n$
	  \ENDWHILE
	\end{algorithmic}
	\end{algorithm}	
	The meaning of variables $c_i,\zeta_i,y_i,\pi_i,d_i,l$ 
	and fixed parameter $\varepsilon>0$ is 
	explained as follows:
	\begin{itemize}
	\item $c_i$ is the {\it price clock} 
	 of buyer $i$, which is used as
	  the transaction cost for 
	  one unit of the goods at this moment.
	  It starts at $0$
	  and increases by
	 $\varepsilon$ in each step.	
	\item $\zeta_i$ is the amount of 
	goods that buyer $i$ clinches in the current iteration.
	\item $y_i$ is the amount of goods of buyer $i$. 
	\item $\pi_{i}$ is the payment of buyer $i$.
	\item $d_i$ is the {\it demand} of buyer $i$, 
	which is interpreted as the maximum possible amount
	 of transactions that buyer $i$ can 
	 get in the future.
	\item $l$ is the buyer whose price clock 
	is increased in the next iteration.
	\end{itemize}
	Algorithm \ref{algo1} terminates 
	when $d_i=0$ for all buyer $i$, and outputs $(y,\pi)$
	at this moment.
	As in Goel et al. \cite{GMP2014} we assume the following:
	\begin{assumption}
	\label{assump_varepsilon1}
	All values of $v_i'$ are multiples of 
	$\varepsilon$.
	\end{assumption}

	We have to explain the detail of 
	``not affecting other buyers" in line 3.
	For a transaction vector~$y$ and 
	demand vector $d$, define 
	the {\em remnant supply polytope} $\tilde{P}_{y,d}$ by
	\begin{equation}
	\label{remnant_supply1}
	\tilde{P}_{y,d}:=\{z\in\mathbb R^{N}_{+}\mid y+z\in \tilde{P},\ z_i\leq d_i\ (i\in N)\}.
	\end{equation}
	Also, for $\zeta\in \mathbb R_+$, define the 
	{\em remnant supply polytope of 
	remaining buyers $N-i$} by 
	\begin{equation}
	\label{remnant_supply_i1}
	\tilde{P}^{i}_{y,d}(\zeta):=\{\sigma\in
	\mathbb R^{N-i}_{+}\mid  
	\exists z\in \tilde{P}_{y,d},\ z_i=\zeta,\ z|_{N-i}=\sigma \}.
	\end{equation}
	The first polytope $\tilde P_{y,d}\subseteq \tilde{P}$ 
	represents the feasible 
	increases of transactions under $y$ and $d$,
	and the second polytope $\tilde P^{i}_{y,d}(\zeta)
	\subseteq \mathbb R^{N-i}_{+}$
	represents the possible  
	amounts of goods that buyers except $i$ can 
	get in the future,
	provided $i$ 
	got $\zeta\in \mathbb R_{+}$ in this iteration.
	For each $i$, clinching  
	$\zeta_i$ in line 3 is chosen as the 
	maximum $\zeta$ for 
	which $\tilde{P}^i_{y,d}(\zeta)=\tilde{P}^i_{y,d}(0)$.
	Then the following theorem holds.
	\begin{theorem}[Goel et al. \cite{GMP2015}]
	\label{goel}
	Algorithm \ref{algo1} satisfies (ICb) and (IRb).
	Moreover, when the algorithm terminates, 
	the transaction vector $y$ belongs to the base polytope of $\tilde{P}$.
	\end{theorem}
	
	Our mechanism makes use of 
	the allocation $(y,\pi)$ for Algorithm \ref{algo1}.
	We transform $(y,\pi)$ into an 
	allocation $(w,p,r)$ for the 
	original two-sided market 
	by using polymatroidal network flow; see Section~\ref{subsec:flow}.
	Let us construct a polymatroidal network $\mathcal N$
	from the market $(N,M,E)$ and $y$; see Figure~\ref{network_1}.
		\begin{figure}[t]
			\centering
			\includegraphics[width=10cm]{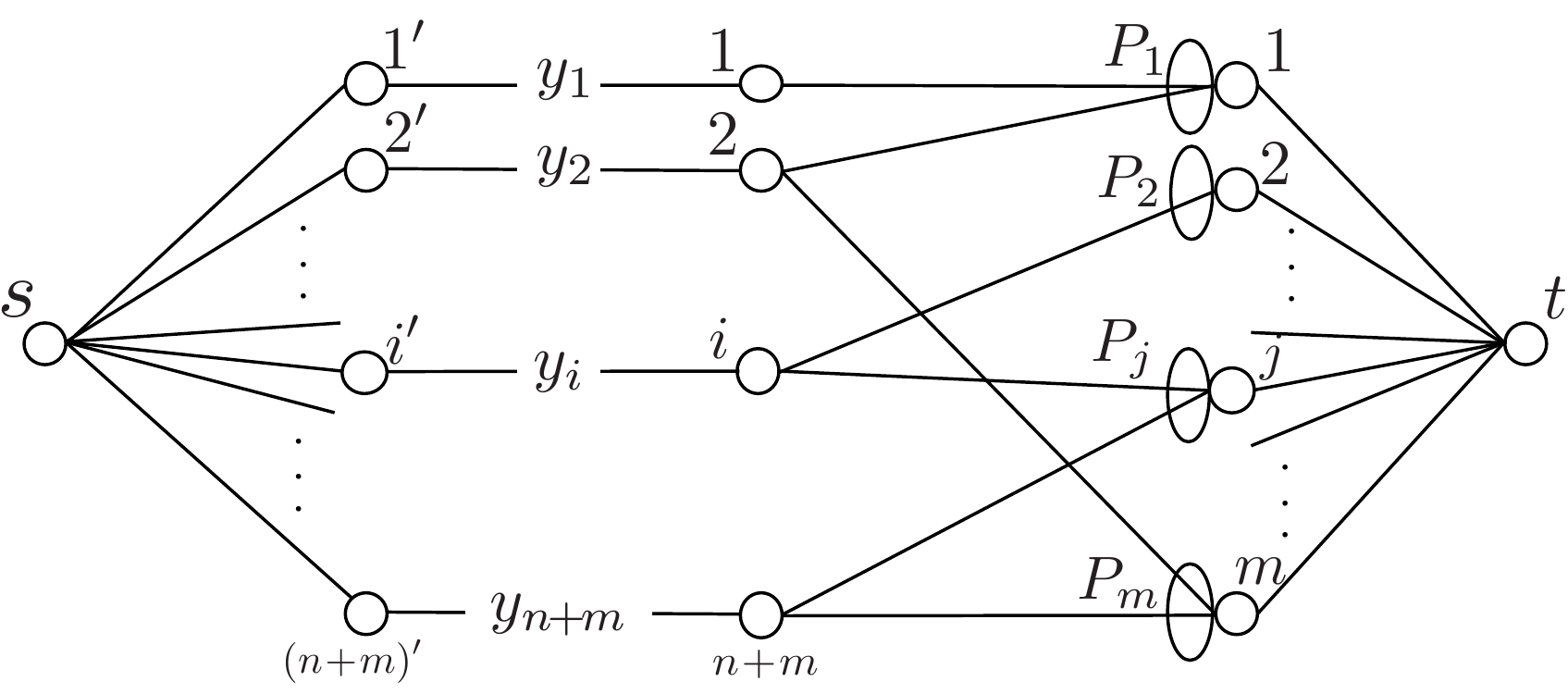}
			\caption{The polymatroidal network $\mathcal N$ in mechanism 1.}
			\label{network_1}
		\end{figure}
	Each edge in $E$ is directed from $N$ to $M$.
	For each buyer $i$, consider copy $i'$ and edge $i'i$.
	The set of the copies of 
	$X\subseteq N$ is denoted by $X'$.
	Add source node $s$ and edge $si'$ for each $i'\in N'$.
	Also, add sink node $t$ and edge $jt$ for each $j\in M$.
	Edge-capacity $c$ is defined from $y$ by 
	\begin{align}
	\label{capacity_mechanism1}
	c(e):=
	\begin{cases}
	y_i\quad \,{\rm if}\ e=i'i \ 
	{\rm for}\ i\in N, \\
	\infty\quad {\rm otherwise}.	
	\end{cases}
	\end{align}
	Two polymatroids of each node are defined as follows.
	For $j\in M$, the polymatroid on $\delta^{-}_{j}=E_j$ is
	defined as $P_{j}$, and 
	the other polymatroid on $\delta^+_{j}$ is 
	defined by the capacity as in (\ref{eqn:capacity}).
	Both two polymatroids of other nodes are defined by
	the capacity. 
	Let $\mathcal N$ denote the resulting 
	polymatroidal network.

	Compute a maximum flow $\varphi$ 
	in $\mathcal N$. The transaction vector $w$ of the 
	original two-sided market is 
	obtained as 
	$w:=\varphi|_E$. 
	 By the 
	 construction, $w$ satisfies the polymatroid constraint. 
	The payment vector $p$  
	is naturally given by $p:=\pi$, which satisfies the budget 
	constraints.
	Also we choose $r$ as an arbitrary vector 
	that satisfies
	\begin{align}
	\label{irs_1}
	r_j - p_{n+j} + \rho_j y_{n+j} & \geq \rho_j f_j(E_j) \quad (j \in M),  \\
		\label{budgetbalance_1}
		\sum_{j\in M}r_j&=\sum_{i\in N}p_i.
	\end{align}
	It turns out that the first condition (\ref{irs_1}) corresponds to (IRs), 
	and the second condition (\ref{budgetbalance_1}) 
	corresponds to~(SBB).	
	\begin{lemma}
	\label{r_nonempty}
	There exists a vector $r$ satisfying $(\ref{irs_1})$ and $(\ref{budgetbalance_1})$.
	\end{lemma}

	We give in the next section a constructive proof of this lemma via mechanism 2.
	 
	After the procedure above, we obtain 
	the allocation $(w,p,r)$ including 
	virtual buyers. 
	We transform this allocation 
	into the allocation in the original market as follows:
	\begin{itemize}
		\item 	For each seller $j$, we regard $y_{n+j} = w_{(n+j)j}$ 
		as the unsold goods of $j$, 
		and replace the revenue $r_j$ by $r_j - p_{n+j}$.
	\end{itemize} 
    By this update, (\ref{irs_1}) and (\ref{budgetbalance_1}) are equal to (IRs) and (SBB), respectively.	
	Then we obtain the following theorem: 
	\begin{theorem}
	\label{mechanism_1}
	Mechanism 1 satisfies all the properties in Theorem \ref{main}.
	\end{theorem}
	\begin{proof} %
	First we prove  
	$w(E_i)=y_i$ for each $i\in N$.
	Since $y$ is obtained by Algorithm \ref{algo1}, 
	it holds $y\in \tilde{P}$ and thus there exists 
	$\tilde{w}:=\{\tilde{w}_{ij}\}_{ij\in E}$ such that 
	$\tilde{w}\in P$ and $\tilde{w}(E_i)=y_i$. 
	Then $\tilde{w}|_{E_j}\in P_j$ for each $j\in M$, and 
	there exists a flow $\tilde{\varphi}$ such that 
	$\tilde{\varphi}|_{E}=\tilde{w}$. 
	By the flow-conservation law, 
	it holds $\tilde{\varphi}({i'i})=
	\tilde{w}(E_i)=y_i$ for each $i\in N$. 
	Since it holds $c(i'i)=y_i$ for each $i\in N$,
	$\tilde{\varphi}$ is a maximum flow in $\mathcal N$.
	Therefore we obtain 
	$\tilde{\varphi}({i'i})=y_i$ for each $i\in N$
	for any maximum flow $\tilde{\varphi}$ in $\mathcal N$.
	Thus the total amount of transaction of buyer $i$
	is certainly $y_i$.
	By this fact and the definition of $p$, 
	(ICb) and (IRb) are obtained immediately from Theorem \ref{goel}.
	 
	 As we mentioned, by the above update of 
	 the goods and revenue for sellers,
	 (IRs) is equivalent to (\ref{irs_1}).
	 Also (SBB) is nothing but (\ref{budgetbalance_1}). 
		
	The proof of (PO) is precisely the same as
	that for mechanism 2 (Theorem~\ref{pareto}), which is given in Section~\ref{sec:proof}.
	\end{proof}

\subsection{Second Mechanism}

	Next we describe mechanism 2,
	which is an extension of the polyhedral clinching 
	auction to two-sided markets.
	First we make the same preprocessing 
	as in mechanism 1 by adding virtual buyers. 

	After the preprocessing, we apply Algorithm 2, which is a 
	two-sided version of the polyhedral clinching auction.
	\begin{algorithm}[htb]
	\caption{Polyhedral\ Clinching\ Auction for Two-Sided Market} 
	\label{algo2}                              
	\begin{algorithmic}[1]
	  \STATE $p_i:=0,\ c_i:=0,\ d_i:=\infty\ 
	  (i\in N)\ {\rm and}\ l:=1$.\\
	  \STATE $r_j:=0\ (j\in M)\ {\rm and} \ w_{ij}:=0\ (ij\in E)$.
	  \WHILE{$d_i\neq 0$ for some $i\in N$}
	  \FOR{$i=1,2,\ldots,n+m$} 
	  \STATE Clinch a maximal increase 
	  $(\xi_{ij})_{ij\in E_i}$ \underline{{\it not affecting other buyers}}.
	  \STATE 
	   $w_{ij}:=w_{ij}+\xi_{ij}\quad (ij\in E_i)$.
	  \STATE 
	  $\displaystyle p_i:=p_i+c_i \xi(E_i).$
	  \STATE $ d_i:=
	 \begin{cases}
	(B_i-p_i)/c_i\quad {\rm if}\  c_i<v'_i,\\
	0 \quad\qquad\qquad\ \  {\rm otherwise}. 
	\end{cases}$
	  \ENDFOR
	  \STATE $\displaystyle r_j:=r_j+\sum_{ij\in E_j} c_i \xi_{ij}\quad (j\in M).$
	  \STATE $c_{l}:=c_{l}+\varepsilon$.
	  \STATE $l:=l+1\ {\rm mod\ } m+n$.
	  \STATE $ d_l:=
	 \begin{cases}
	(B_l-p_l)/c_l\quad {\rm if}\  c_l<v'_l,\\
	0 \quad\qquad\qquad\ \  {\rm otherwise}.
	\end{cases}$
	 \ENDWHILE
	\end{algorithmic}
	\end{algorithm}
	The variables $c_i,d_i,l$ and the parameter $\varepsilon$ 
	are the same as in Algorithm \ref{algo1}. 
	Suppose that $\varepsilon$ satisfies 
	Assumption \ref{assump_varepsilon1}.
	The meaning of variables $\xi_{ij},w_{ij},p_i,r_j$ 
	is explained as follows: 
	\begin{itemize}
	\item $\xi_{ij}$ is the increase of 
	transactions between buyer $i$ and seller $j$ in the current
	iteration.
	\item $w_{ij}$ is the total amount of 
	transaction between buyer $i$ and seller $j$.
	Transaction vector $(w_{ij})_{ij\in E}$ must 
	belong to the polymatroid $P$ (defined in (\ref{eqn:P})).
	\item $p_{i}$ is the payment of buyer $i$.
	\item $r_{j}$ is the revenue of seller $j$. 
	\end{itemize}
	
	Algorithm \ref{algo2} terminates 
	when $d_i=0$ for all 
	buyers $i$, and outputs $(w,p,r)$
	at this moment.
	The allocation $(w,p,r)$ obtained by Algorithm \ref{algo2} 
	includes that for virtual buyers. 
	Applying the same transformation as in 
	mechanism 1, we obtain the allocation 
	$(w,p,r)$ for the original market.
		
	Again we explain the details of ``not affecting
	other buyers" in line 5.
	For transaction vector $w$ and 
	demand vector $d:=(d_i)_{i\in N}$,
	define {\it the remnant supply polytope} $P_{w,d}$ 
	by 
	\begin{equation}
	\label{remnant_supply}
	P_{w,d}:=\{x\in\mathbb R^{E}_{+}\mid 
	w+x\in P,\ x(E_i) \leq d_i\ (i\in N)\}.
	\end{equation}
	In addition, for 
	$\xi:=(\xi_{ij})_{ij\in E_i}\in \mathbb R^{E_i}_{+}$,
	define {\it the remnant supply polytope $P^{i}_{w,d}(\xi)$ of 
	remaining buyers  $N-i$} by
	\begin{equation}
	\label{remnant_supply_i}
	P^{i}_{w,d}(\xi):=\{u\in
	\mathbb R^{N-i}_{+}\mid  
	\exists x\in P_{w,d},\ x|_{E_i}=\xi,\  x(E_{k})=u_k\ (k\in N-i )\}.
	\end{equation}
	The first polytope $P_{w,d}\subseteq P$ 
	represents the feasible 
	increases of transactions under $w$ and $d$,
	and the second polytope $P^{i}_{w,d}(\xi)
	\subseteq \mathbb R^{N-i}_{+}$
	represents the possible  
	amounts of goods that buyers except $i$ can 
	get in the future,
	provided $i$ 
	got $\xi\in \mathbb R^{E_i}_{+}$ in this iteration.
	Then the condition for clinching $\xi:=(\xi_{ij})_{ij\in E_i}$
	not to affect other buyers
	is naturally written as $P^i_{w,d}(\xi)=P^i_{w,d}(0)$.
	This motivates to define the polytope $P^{i}_{w,d}$ by 
	\begin{equation}\label{eqn:clinching_polytope}
	P^{i}_{w,d}:=\{\xi\in \mathbb R_{+}^{E_i}\mid P^i_{w,d}(\xi)=P^i_{w,d}(0)\},
	\end{equation}
	 which we call the 
	{\it clinching polytope} of buyer $i$.
	A clinching vector $(\xi_{ij})_{ij\in E_i}$ in line 5 is 
	chosen as a maximal vector in the clinching 
	polytope $P_{w,d}^{i}$.
	The following theorem says that
	this clinching step can be done in polynomial time. 
	\begin{theorem}
	\label{greedy}
	The clinching polytope $P^{i}_{w,d}$ is a polymatroid, 
	and the value of the 
	corresponding submodular function can be computed 
	in polynomial time.
	In particular,  by the greedy algorithm, a maximal vector of $P^{i}_{w,d}$ 
	can be computed in 
	polynomial time.
	\end{theorem}	  
	This implies that our mechanism 
	is implementable in practice. 
	The proof of Theorem \ref{greedy} is given in
	Section~\ref{subsec:greedy}. 
		
	We analyze mechanism 2.
	First we show (SBB) for our mechanism.
	In line 10 of Algorithm \ref{algo2}, 
	the total payment of buyers (including virtual buyers) is 
	directly given to sellers in each iteration. 
	Also, in the transformation of the allocation,  
	both the total payment and the total revenue are subtracted by
    the total payment of virtual buyers.
	Thus we have: 
	\begin{lemma}\label{lem:SBB_Algo2}
	Mechanism 2 satisfies (SBB).	
	\end{lemma}
	For the individual rationality of sellers (IRs), 
	we will prove in Section~\ref{subsec:greedy} that
	each seller $j$ and 
	nonvirtual buyer $i$ transact 
	only when the price clock $c_i$ 
	is at least valuation $\rho_j$ and virtual buyer $n+j$ 
	vanishes, i.e., $d_{n+j}=0$.
	Consequently the goods of $j$ are sold at price at least $\rho_j$, and we obtain (IRs); see (\ref{equivalent_to_IRs}).
	\begin{proposition}
	\label{IRs}
	Mechanism 2 satisfies (IRs). 
	\end{proposition}
    The proof is given in Section~\ref{subsec:greedy}.
	Next we consider other properties (ICb),(IRb), 
	and (PO) in Theorem \ref{main}.
	We utilize the following theorem 
	on the relationship between mechanism 1 and mechanism 2.
	\begin{theorem}
	\label{relation}
	Suppose that $\varepsilon$ is the same in both 
	Algorithms.
	At the end of algorithm, $(y,\pi)$ in Algorithm \ref{algo1}  
	is equal to 
	$((w(E_i))_{i\in N},p)$ 
	in Algorithm \ref{algo2}.
	\end{theorem}
	The proof of Theorem \ref{relation} is 
	given in Section~\ref{subsec:relation}. 
	Therefore in mechanism 2, the utility of buyers 
	is the same as that in the reduced one-sided market.
	By this fact and Theorem \ref{goel} we have:
	\begin{corollary}
	\label{ICb}
	Mechanism 2 satisfies (ICb) and (IRb).
	\end{corollary}
	
	We now show Lemma~\ref{r_nonempty} in mechanism 1.
	\begin{proof}[Proof of Lemma \ref{r_nonempty}]
		We verify that the revenue vector $r$ obtained by Algorithm \ref{algo2}
		satisfies (\ref{irs_1})(=(IRs)) and (\ref{budgetbalance_1})(=(SBB)) for mechanism 1.
		By Theorem \ref{relation}, the payment vectors of buyers in both algorithms are the same.
		Therefore (\ref{budgetbalance_1}) follows from Lemma~\ref{lem:SBB_Algo2}. 
		Also, by Theorem~\ref{relation}, 
		the amounts of unsold goods, which are the amounts of goods obtained by virtual buyers, 
		in both algorithms are the same.
		Hence (IRs) for mechanism 2 (Proposition~\ref{IRs}) implies (\ref{irs_1}) = (IRs) for mechanism 1.
	\end{proof}

	Goel et al. \cite{GMP2015} proved (PO) for Algorithm \ref{algo1}. 
	We combine Theorem 3.4 with their proof method 
	and careful considerations on utilities of sellers,
	and prove (PO) for mechanism 2:
	\begin{theorem}
	\label{pareto}
	Our mechanism satisfies (PO).
	\end{theorem}
	The proof is given in Section~\ref{subsec:pareto}.

	\begin{remark}[Concave Budget Constraint]\label{rem:concave}
		{\rm 
	 Our model and mechanisms 
	 naturally incorporate concave budget constraints, i.e., 
	 utility $u_i$ of each buyer $i$ is given by
	\begin{equation*} 
	u_i(\mathcal A):=
	\begin{cases}
	\displaystyle v_i w(E_i) -p_i\ \ \ {\rm if}\ p_i\leq \phi_i( w(E_i)), \\
	-\infty \qquad\quad\qquad \ \ {\rm otherwise},
	\end{cases}
	\end{equation*}
	where $\phi_i: \mathbb R_+\to \mathbb R_+$ is
	a concave non-decreasing 
	 function with $\phi_i(0)=0$.
	 In the case of the usual budget constraint, 
	 $\phi_i$ is defined by $\phi_i(0) := 0$ and $\phi_i(x) := B_i$ for $x > 0$.
	Concave budget constraint was considered 
	by Goel et al. \cite{GMP2014} for one-sided markets.
	They showed that the polyhedral 
	clinching auction 
	can be generalized to this setting. 
	By using their idea, our mechanisms can also be adapted to have the 
	same conclusion in Theorem \ref{main}.
	
	The adaptation is given as follows.
	Each buyer reports the region 
	$\{(p_i, z_i) \,|\,  p_i \leq \phi_i(z_i)\}$ 
	to the auctioneer, and the budget 
	feasibility is replaced by  
	$p_i \leq \phi_i(w(E_i))$.
	Assume, as in Goel et al. \cite{GMP2014}, that 
	the interval $\varepsilon$ of price clocks
	divides $\beta_i:=\lim_{x\to 0} \phi_i(x)/x$ 
	for each $i$.
    The function $\phi_i$ for a virtual buyer $i\in N$ 
    is defined by $\phi_i(z_i) := 0$ if $z_i = 0$ and $\infty$ otherwise.   
	The demand $d_i$ is defined by	
	\begin{equation*}
	\label{concave_demand1}
	d_i:=  
	\left\{
	\begin{array}{ll}
	\max\{z_i\mid \pi_i+c_i z_i\leq \phi_i(y_i+z_i) \} & {\rm if}\ c_i < v'_i, \\
	0 &{\rm otherwise}.
	\end{array}
	\right.
	\end{equation*}
	In mechanism 2, we use $w(E_i)$ instead of $y_i$,
	and $p_i$ instead of $\pi_i$.
	In Section~\ref{sec:proof}, we give every proof for the setting with 
	concave budget constraints. 	
	}
\end{remark}

\subsection{Example}\label{subsec:example}
Here we demonstrate the behavior of the mechanisms 
for a small two-sided market consisting of two buyers $1,2$ 
and two sellers $1,2$.
Each seller $j$  has $s_j$ units of goods, 
and can transact with both buyers $1,2$ freely.
Namely the market is represented as a complete bipartite graph $(N,M,E)$ with
$N = \{1,2\}$, $M = \{1,2\}$ and $E = N \times M$, 
and the polymatroid constraint $P_j$ $(j=1,2)$ 
is given by	the stock constraint
$w_{1j} + w_{2j} \leq s_j$, $w_{1j} \geq 0$, and $w_{2j} \geq 0$; the corresponding submodular function $f_j$ 
is given by $f_j(F) := s_j$ if $F \neq \emptyset$ and $f_j(\emptyset) := 0$.
Suppose that buyer $i$ has budget $B_i$ and bid $v'_i = v_i$, 
and seller $j$ has reserved price $\rho_j$.

The both mechanisms 1 and 2 first apply the preprocessing.
Add virtual buyers $3,4$ to $N$, 
where $3$ and $4$ are adjacent only to sellers $1$ and $2$, 
respectively. 
Their bids and budgets are defined by
$v_3 := \rho_1$, $v_4 := \rho_2$, and
$B_3 = B_4 := \infty$.
According to (\ref{f_j}),
the polymatroid constraint $P_j$ $(j=1,2)$ is extended to
\begin{eqnarray*}
&& w_{1j} + w_{2j} + w_{(j+2) j} \leq  s_j, \\
&& w_{1j},w_{2j}, w_{(j+2)j}  \geq  0.
\end{eqnarray*}
Then the whole polymatroid $P = P_1 \oplus P_2$ is obtained by the union of 
these inequalities for $j=1,2$. See the left of Figure~\ref{fig:2x2}.
	\begin{figure}[t]
		\centering
		\includegraphics[width=12cm]{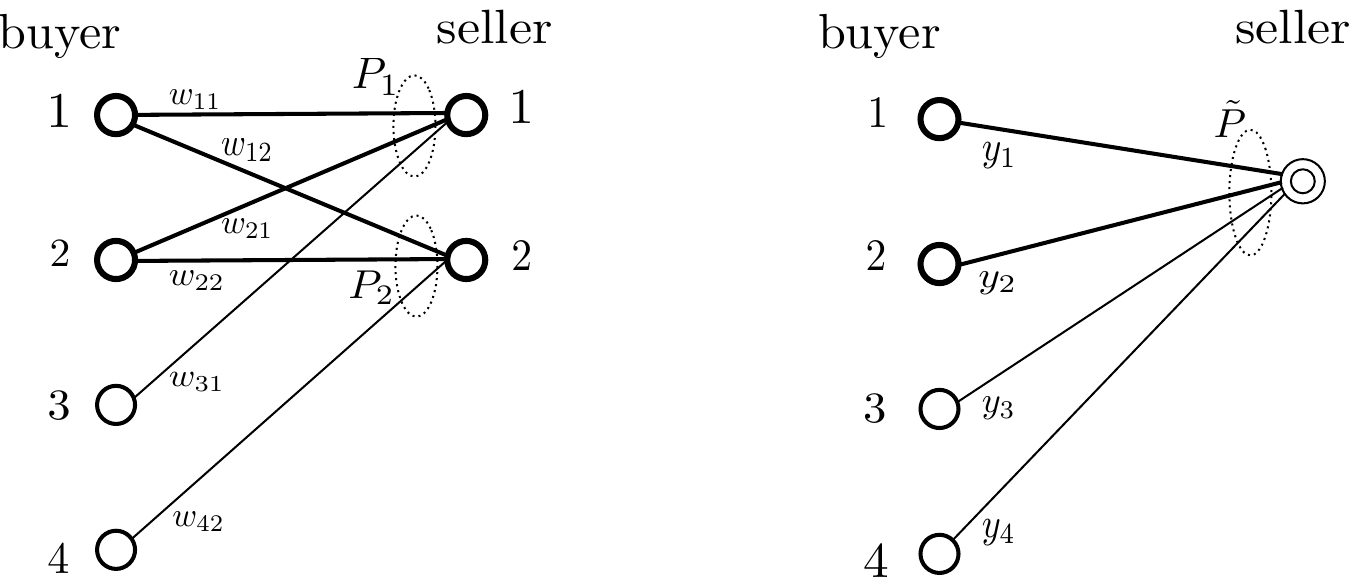}
		\caption{A two-sided market (left) and the reduced one-sided market (right)}
		\label{fig:2x2}
	\end{figure}

We first consider mechanism 2.
In some iteration of the mechanism, 
suppose that
the transaction and demand vectors are given by
$w$ and $d$, respectively.  
Let $\tilde s_j$ denote the current stock of the goods 
of seller $j$, i.e., $\tilde s_j := s_j - w_{1j} - w_{2j} - w_{(j+2)j}$.
Then the remnant supply polytope $P_{w,d}$ is given by
\begin{eqnarray*}
&&	 x_{1j} + x_{2j} + x_{(j+2) j}  \leq  \tilde{s}_j \quad (j=1,2),  \\
&&	 x_{i1} + x_{i2}  \leq d_i \quad (i=1,2), \\
&&	 x_{(j+2)j}  \leq  d_{j+2} \quad (j=1,2), \\
&&  x_{ij} \geq 0 \ (ij \in E). 
\end{eqnarray*}
We say that virtual buyer $j+2$ {\em exists} if $d_{j+2} = \infty$ 
and {\em vanishes} if $d_{j+2} = 0$.  
For $j \in \{1,2\}$, let $e_j \in \{0,1\}$ be defined by 
$e_j := 1$ if $j+2$ exists and $e_j := 0$ otherwise.
Also let $\bar e_j := 1 - e_j$.
For $i \in \{1,2\}$, let $i'$ denote the companion of $i$, i.e., $\{1,2\} = \{i,i'\}$.
For nonvirtual buyer $i \in \{1,2\}$ and $\xi \in \mathbb{R}_+^{E_{i}}$, 
the remnant supply polytope $P^{i}_{w,d}(\xi)$ of 
remaining buyers  $\{i',3,4\}$ is given by
\begin{eqnarray*}
&& u_{i'} + u_3 + u_4  \leq  \tilde {s}_1 + \tilde{s}_2 - \xi_{i1} - \xi_{i2},  \\
&& u_{i'}  \leq  d_{i'},\ u_3  \leq  (\tilde{s}_1 - \xi_{i1})  e_{1},\ u_4  \leq  
(\tilde{s}_2 - \xi_{i2}) e_2, \\
&& u_{i'}, u_3,u_4  \geq  0
\end{eqnarray*}
if $\xi_{i1} + \xi_{i2} \leq d_i$, and $P^{i}_{w,d}(\xi) = \emptyset$ otherwise.
For virtual buyer $j+2$ with $j \in \{1,2\}$ (if it exists) 
and $\xi = \xi_{(j+2)2} \geq 0$, 
the polytope $P^{j+2}_{w,d}(\xi)$ is given by
\begin{eqnarray*}
&& 	u_{1} + u_2 + u_{j'+2}  \leq  \tilde {s}_1 + \tilde{s}_2 - \xi_{(j+2)j},  \\
&&	u_{1}  \leq  d_{1},\ u_2  \leq  d_2,\ u_{j'+2}  \leq  \tilde{s}_{j'} e_{j'}, \\
&&	u_{1}, u_2,u_{j'+2}  \geq  0
\end{eqnarray*}
if $\xi_{(j+2)j} \leq \tilde{s}_j$, and $P^{j+2}_{w,d}(\xi) = \emptyset$ otherwise.
By elementary calculation, we can determine
clinching polytopes $P_{w,d}^i$ $(i=1,2,3,4)$ as follows.
\begin{itemize}
	\item 
For nonvirtual buyer $i \in \{1,2\}$, the clinching polytope $P^{i}_{w,d}$ is given by:
\begin{eqnarray*}
	&& \xi_{i1} + \xi_{i2}  \leq \min \{d_i, \max \{ \tilde {s}_{1}\bar e_1 +  \tilde{s}_2\bar e_2 - d_{i'},0\} \}, \\
	&& 0 \leq 	\xi_{i1} \leq  \tilde{s}_1\bar e_1,\ 0 \leq 	\xi_{i2} \leq  \tilde{s}_2\bar e_2.
\end{eqnarray*}
\item For virtual buyer $j+2$ with $j \in \{1,2\}$, the clinching polytope $P^{j+2}_{w,d}$ is given by 
\[
0 \leq \xi_{(j+2)j} \leq \max \{\tilde{s}_{j} + \tilde{s}_{j'} \bar e_{j'}  -d_1-d_2,0 \}.
\]
\end{itemize}
If virtual buyer $j+2$ exists, then 
any nonvirtual buyer cannot transact with $j$, 
i.e., $\xi_{ij} = 0$,
since this will interrupt $j+2$ who has infinite demand.
After virtual buyers vanish, buyer $i$ 
can  clinch surplus $\tilde s_1 + \tilde s_2 - d_{i'}$
which exceeds the demand of the competitor $i'$.
Here maximal clinching vectors
forms the segment between $(\min \{ \tilde{s}_1,c\}, c - \min \{\tilde{s}_1,c\})$ 
and $(c - \min\{ \tilde{s}_2,c\}, \min \{  \tilde{s}_2,c\})$, where
$c := \min\{ d_i, \max \{\tilde {s}_{1} + \tilde{s}_2 - d_{i'},0 \}\}$.
The clinch by virtual buyer $j+2$ can be understood as  
taking back the goods of $j$ exceeding  
the limit that nonvirtual buyers can buy.

The auction proceeds as follows.
In earlier iterations where a virtual buyer $j+2$ exists, 
nonvirtual buyers cannot purchase the good of $j$.
Also virtual buyer $j+2$ may clinch to 
take back the goods of $j$ which is guaranteed to be unsold. 
The virtual buyer $j+2$ will vanish when the price clock $c_j$ 
reaches at the reserved price $\rho_j$ of $j$ ($=$ bid $v_{j+2}$ of $j+2$).
Consequently the good of $j$ will be sold at price at least $\rho_j$. 
In this way, virtual buyers increase 
the prices of the goods to reserved prices, 
and take back the unsold goods, until their vanishing.
After that,  
nonvirtual buyer $i$ clinches surplus $\tilde s_1 + \tilde s_2 - d_{i'}$
exceeding the potential of competitor $i'$.
Thus the total sum of demands $d_1 + d_2 + d_3 + d_4$
keeps at least the current stock  $\tilde s_1 + \tilde s_2$.
Prices are increasing, and demands are decreasing.
Finally all demands are zero, and the stocks are zero, i.e., all goods are sold.
The auction terminates.

Table~\ref{tab:behavior} shows the behavior of mechanism 2 for this market with 
setting $v_1 = 3$, $B_1 = 12$, $v_2 = 3$, $B_2 = 11$, $\rho_1 = 1$, $s_1 = 7$, $\rho_2 = 1$, and $s_2 = 8$.
\begin{table}[t]
	\center
	\caption{behavior of mechanism 2}
	\begin{tabular}{|c|cccc|cccc|ccc|ccc|}
		\hline
		&  \multicolumn{4}{c|}{buyer 1} &\multicolumn{4}{c|}{buyer 2} & \multicolumn{3}{c|}{seller 1} &   \multicolumn{3}{c|}{seller 2} \\
		&  \multicolumn{4}{c|}{$B_1 =12$, $v_1=3$} &\multicolumn{4}{c|}{$B_2 =11$, $v_2 = 3$} & \multicolumn{3}{c|}{$\rho_1 = 1$, $s_1=7$} &   \multicolumn{3}{c|}{$\rho_2 = 1$, $s_2 = 8$} \\
		\hline \hline
		$l$ & $c_1$ & $d_1$ & $\xi_{11}$ & $\xi_{12}$ & $c_2$ & $d_2$ & $\xi_{21}$ & $\xi_{22}$ & $\tilde s_1$ & $\varDelta r_1$ & $c_3$ & $\tilde s_2$ & $\varDelta r_2$ & $c_{4}$ \\
		\hline  
		$1$  & $0$ & $\infty$&  &        & 0  & $\infty$ &      &    & $7$  &   & 0 & $8$ &  &   0  \\
		$2$  & $1$ & $12$  &    &        & 0  & $\infty$ &      &    & $7$  &   & 0 & $8$   &  & 0 \\
		$3$  & $1$ & $12$  &    &        & 1  & $11$     &      &    & $7$  &   & 0 & $8$   &  & 0 \\
		$4$  & $1$ & $12$  &    &        & 1  & $11$     &      &    & $7$  &   & 1 & $8$   &  & 0 \\
		$5$  & $1$ & $12, 8$  & $2$& $2$    & 1  & $11$    & $3/2$&$3/2$& $7, 5$ & $7/2$& 1 & $8, 6$  & $7/2$ & 1 \\
		$6$  & $2$ & $4$  &     &        & 1    & $8$   & $7/4$&$9/4$& $7/2$ & $7/4$  & 1 & $9/2$ & $9/4$ & 1 \\
		$7$  & $2$ & $4,2$  & $7/8$&$9/8$& 2   & $2$  &       &      &  $7/4,7/8$ & $7/4$& 1 & $9/4, 9/8$& $9/4$  & 1 \\
		$8$  & $2$ & $2$  &     &         & 2    & $2$    &      &      & $7/8$ &   & 2 & $9/8$ &  & 1 \\
		$9$  & $2$ & $2$  &     &         & 2    & $2$    &      &      & $7/8$  &   & 2  & $9/8$ &  & 2 \\ 
		$10$ & $3$ & $0$  &     &         & 2    & $2$    & $7/8$  & $9/8$ & $7/8$  & $7/4$  & 2 & $9/8$ & $9/4$ & 2 \\
		\hline \hline
		&  \multicolumn{4}{l|}{obtained goods: $6$} &\multicolumn{4}{l|}{obtained goods: $9$} & \multicolumn{3}{l|}{unsold goods: $0$} &   \multicolumn{3}{l|}{unsold goods: $0$} \\
		&  \multicolumn{4}{l|}{payment: $8$} &\multicolumn{4}{l|}{payment: $11$} & \multicolumn{3}{l|}{revenue: $8.75$} &   \multicolumn{3}{l|}{ revenue: $10.25$} \\
		&  \multicolumn{4}{l|}{utility: $10$ } &\multicolumn{4}{l|}{utility: $16$} & \multicolumn{3}{l|}{utility: $8.75$ } &  \multicolumn{3}{l|}{utility: $10.25$} \\
		\hline 
	\end{tabular}
	\label{tab:behavior}
\end{table}
In the first four iterations, nonvirtual 
buyers cannot clinch (i.e., $\xi_{ij} = 0$), 
since price clocks are less than reserved prices of sellers, i.e., 
virtual buyers exist.
Also no taking back by virtual buyers occurs;  
hence clinching vectors $\xi_{31}$ and $\xi_{42}$ are omitted in the table.
The first clinching occurs in the 5th iteration ($l=5$).  
The clinching polytope $P_{w,d}^1$ for buyer~$1$ is 
given by $\xi_{11} + \xi_{12} \leq \min \{ 12, \max \{7+8 - 11,0 \}\} = 4$, 
$\xi_{12} \leq 8$, and $\xi_{12} \leq 7$.
The maximal clinching vectors form the 
segment between $(4,0)$ and $(0,4)$.
Here the midpoint $(2,2)$ is chosen as a clinching vector $(\xi_{11},\xi_{12})$; 
this midpoint clinching rule is used in the sequel.
Buyer $1$ pays $2$ to both sellers $1$ and $2$.
The demand $d_1$ of buyer $1$ is updated to $8$.
The stocks of sellers $1$ and $2$ 
are updated to $5$ and $6$, respectively.
The next is the turn of buyer $2$, who also 
clinches $(\xi_{21}, \xi_{22}) = (3/2,3/2)$.
In this iteration, sellers $1$ and $2$ obtain revenues 
$\varDelta r_1 = \varDelta r_2  = 7/2$. 
The auction proceeds in this way.
The result is given in the bottom of the table.

Next we consider mechanism $1$.
The reduced one-sided market is obtained as follows. See the right of Figure~\ref{fig:2x2}. 
The aggregated polymatroid $\tilde P$ is given by
\begin{eqnarray*}
&& y_1 + y_2 + y_3 + y_4 \leq s_1 + s_2, \\
&& y_3 \leq s_1,\ y_4 \leq s_2, \\
&& y_1,y_2, y_3,y_4 \geq 0.
\end{eqnarray*}
Given a transaction $y$ and demand $d$, 
the remnant supply polytope $\tilde P_{y,d}$ is given by
\begin{eqnarray*}
	&& z_1 + z_2 + z_3 + z_4 \leq s', \\
	&& z_1 \leq d_1,\ z_2 \leq d_2,\ z_3 \leq s'_1 e_1,\ z_4 \leq s'_2 e_2 \\
	&& z_1,z_2, z_3,z_4 \geq 0, 
\end{eqnarray*}
where $s' := s_1 + s_2 - y_1 - y_2 - y_3 -y_4$ and $s'_j := s_j - y_{j+2}$ for $j=1,2$.
Via $\tilde P_{y,d}^{i}(\zeta)$, the maximum clinch $\zeta_i$ 
is computed as follows:
\begin{itemize}
	\item $\zeta_i =  \min \{d_i, \max \{ s' - s'_1e_1 - s'_2 e_2 - d_{i'},0\}\}$ for $i=1,2$.
\item
$\zeta_{j+2} =  \min \{ s'_{j}, \max \{ s' -  s'_{j'}e_{j'} - d_{1} - d_{2},0 \}\}$ for $j=1,2$.
\end{itemize}
Then the auction proceeds as follows.
Buyer 1 obtains $4$ units of goods at $l=5$, and $2$ units at $l=7$. 
The total is $6$ and the payment is $8$.
Buyer 2 obtains $3$ units at $l=5$, $4$ units at $l=6$, and $2$ units at $l=10$.
No clinching by virtual buyers occur.
For each buyer, 
the obtained goods and the payment 
are the same as that in mechanism 2, as Theorem~\ref{relation} says.
In fact, a stronger property holds: 
they are the same at each iteration, i.e., 
$\zeta_i = \xi_{i1} + \xi_{i2}$ holds.
This property is generally true, 
which will be shown in Corollary~\ref{clinch_limit} and Theorems~\ref{clinch_goel} and \ref{f=g}.
After the auction, transaction $(w_{ij})$ 
is recovered so that it satisfies 
$w_{11} + w_{12} = 6$, $w_{21} + w_{22} = 9$, 
$w_{11} + w_{21} =7$, and $w_{12} + w_{22} = 8$.
One example is $(w_{11},w_{12},w_{21}, w_{22}) = (3,3,5,4)$.
The whole revenue of the market is $19$, which is distributed to sellers $1,2$ 
so as to satisfy (\ref{irs_1}) and (\ref{budgetbalance_1}),
i.e., $r_1 + r_2 = 19$, $r_1 \geq 7$, and $r_2 \geq 8$.   
One possible distribution is $(r_1,r_2) = (9,10)$, i.e., 
both sellers have net profit~$2$.

\section{Discussion}\label{sec:discussion}
In this paper, we presented the first generic framework (model and mechanisms) 
for auctions in two-sided markets under budget constraints. 
We think that our framework will be a springboard    
toward algorithmic mechanism design on two-sided markets. 
Here we discuss issues of our framework, 
and raise future research directions.  

Our model assumes, for avoiding the impossibility theorem,  
that all sellers are truthful.
This results in that our framework has a strong priority on buyers.
Indeed, as shown in Theorem~\ref{relation}, 
our mechanisms yield the same output for buyers 
as the original polyhedral clinching auction 
applied to the reduced one-sided market.
From this point, our mechanisms are acceptable to the buyer side.
Therefore our mechanisms should be discussed from the seller side.
For sellers, we consider only the minimum requirement (IRs), 
and do not take into account which sellers should 
obtain more revenue than others.
The total revenue can be much greater than 
the total sum of the guaranteed minimum revenues of sellers.
Thus, how the surplus is shared to sellers is one of central issues.  

Mechanism 1 aggregates sellers into one seller at first.
This makes the auction simple
but loses the correspondence between goods and sellers.
In particular, when a buyer $i$ clinches 
an amount $\zeta_i$ of the goods at price $c_i$, 
the resulting revenue $c_i \zeta_i$ is not distributed at this moment.
In the last recovery step, the total revenue is distributed to sellers so that 
only (IRs) and (SBB) are satisfied.
This can cause some unfairness.
In the example in Section~\ref{subsec:example}, 
the revenue sharing $(r_1,r_2) = (7,12)$ or $(11,8)$ 
is possible, which is apparently unfair 
since the two sellers have almost the same characteristics.

Another type of unfairness occurs in the following extreme example.
The market consists of two buyers $1,2$ and two sellers $1,2$, where 
both sellers have one unit of divisible goods ($s_1 = s_2 = 1$) 
and zero reserved price ($\rho_1 = \rho_2 = 0$). 
The buyer $1$ has budget $B_1 = 1$ and bid $v_1' = 2$,  
and  wants the goods of both sellers, 
The buyer $2$ has infinite budget $B_2 = \infty$ and bid $v_2' = 1$, 
and wants the goods of seller $2$ only.
In particular the goods of seller $2$ are more competitive than that of $1$.
The market is depicted in Figure~\ref{fig:2x2-1}.
	\begin{figure}[h]
		\centering
		\includegraphics[width=7cm]{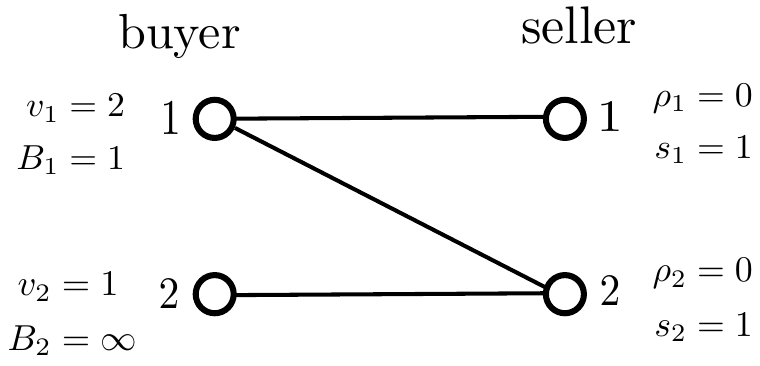}
		\caption{Seller 2 should obtain more revenue than seller 1.}
		\label{fig:2x2-1}
	\end{figure}
Apply Algorithm \ref{algo1} with $\varepsilon=1$; no virtual buyer exists by $\rho_1 = \rho_2 = 0$. 
At the first iteration, buyer 1 clinches
$1$ unit at price $0$, and also clinches 
$1$ unit at price $1$ when buyer 2 drops out from the auction, 
i.e., the price clock of buyer 2 reaches his bid. 
Then all the goods are sold.
By (\ref{irs_1}) and (\ref{budgetbalance_1}), 
the revenue is distributed arbitrarily so that 
it satisfies $r_{1}\geq 0$, $r_{2}\geq 0$, $r_{1}+r_{2}=1$.
In particular, $(r_1,r_2)=(1,0)$ can be chosen.
However this revenue sharing opposes 
to our intuition that sellers who sold competitive goods 
should gain more revenue.
Seller $2$ obtains no revenue, though his goods are competitive; 
it is wanted by both buyers.
On the other hand, seller $1$ obtains revenue $1$ but 
his goods were sold by buyer 1 without competition. 
This example suggests that we should
consider further constraints in 
the last recovery step
so that the resulting sharing is ``fair."
However we found it rather difficult to 
formalize the fairness mathematically.  
Such a concept should be defined only from 
the available information at this step, e.g., 
network structure $(N,M,E)$ and 
the recovered transaction $w$ (possibly not unique).
This is a highly nontrivial problem, 
and should be left to future work.

Mechanism 2 is free from this annoying revenue sharing procedure.
Indeed, actual transactions, prices, and revenues are opened 
and dynamically determined through iterations of the auction.
Contrary to mechanism 1,  
the revenue obtained in each clinching is shared at this moment. 
Clearly sellers who sold their goods in high prices obtain more revenue. 
This is reasonable for all players in the auction.
Also we think that mechanism 2 takes into account 
the competitiveness of goods by the following reason. 
Suppose that a buyer has a link to a competitive seller. 
Then the transaction on this link is influential to others, 
and hence it is hard for the buyer to clinch a large amount of the goods 
through this link.
This consequently makes competitive goods sold in high prices, 
since the price increases through the auction.
If mechanism 2 is applied to the above example, 
the resulting revenue of sellers is $(r_1,r_2)=(0,1)$, which is consistent with 
our intuition. 
From this point, we can say that
mechanism 2 is superior to mechanism 1.

The main issue of mechanism 2 is 
the choice of maximal clinching vectors 
and its affect on the revenue of sellers.
The clinching polytope is a polymatroid, 
and a maximal clinching vector is efficiently obtained (Theorem~\ref{greedy}) but it is not unique in general.
The choice of a maximal clinching vector is directly related to 
the revenues of sellers. 
Indeed, 
nontrivial clinching occurs after virtual buyers vanish, 
and the prices are monotonically increasing in the auction. 
This means that sellers obtain more revenue if his goods are sold in latter iterations.
In the example of Section~\ref{subsec:example}, 
we used the midpoint point clinching rule. 
We can change the clinching rule as: 
(i) each buyer $i$ chooses the clinching vector with the maximum $\xi_{i1}$
and (ii) each buyer $i$ chooses the clinching vector with the maximum $\xi_{i2}$.
Then the revenue vector $(r_1,r_2)$ 
changes from $(8.75,10.25)$ 
to $(7,12)$ for (i) and to $(11,8)$ for (ii). 
Namely, by the choice of clinching vectors, 
mechanism 2 can yield the unfair revenue sharing $(7,12)$ and $(11,8)$.

Therefore we should 
consider a ``fair'' rule of choosing a maximal clinching vector.
One practically reasonable way
is to choose a random ordering of sellers (linked to a buyer doing clinch)  
and to run the greedy algorithm according to the ordering.
Another natural way is to choose a kind of ``center'' 
in the polytope (base polytope)
of maximal vectors of the clinching polytope. 
Natural candidates are the Shapley vector and 
a vector obtained by solving the resource allocation problem with 
an appropriate objective function.
The former is NP-hard to be computed, and the latter   
is efficiently obtained via nonlinear optimization technique over polymatroids~\cite[Chapter V]{F2005}.
However we do not know any theoretical guarantee for these ways of clinching, 
even we did not formulate 
appropriate mathematical concepts (beyond (IRs)) that represent 
the fairness for sellers and that should be guaranteed.

Formalizing such a fairness concept 
may be an unavoidable theme for two-sided models.
One possible approach is 
the {\em $\alpha$-envy freeness} introduced by Goel et al.~\cite{GLMNP2016}, 
which is explained as follows.
In the usual situation of auction, 
each seller $j$ obtains revenue via {\em buying events} 
involving $j$. Here let  ${\cal E}_j$ denote the set of buying events involving $j$.
Then the revenue sharing of a mechanism 
is said to be {\em $\alpha$-envy-free} $(\alpha \leq 1)$
if for any pair of sellers $j,j'$,
the revenue of $j$ (obtained at ${\cal E}_j$) 
is at least $\alpha$ times the (scaled) revenue of $j'$ obtained at ${\cal E}_{j'} \cap {\cal E}_{j}$. 
The $1$-envy free situation achieves an ideal fair sharing  
in the sense that 
every seller $j$ does not envy 
any other seller for his revenue obtained in events that $j$ participates.
In mechanism 2, 
a buying event is precisely the clinch of some buyer $i$ at some iteration $l$.
Namely ${\cal E}_j$ is the set of all pairs $(i,l)$ of buyers $i \in E_j$
and positive integers $l$.
Let $r_j(i,l)$ denote the revenue of seller $i$ 
obtained at the clinch of buyer $i$ in iteration $l$.
Then the $\alpha$-envy freeness in mechanism 2 can 
be formulated:
For any pair of sellers $j,j'$,  
it holds  
\[
r_j = 
\sum_{(i,l) \in {\cal E}_j} r_j(i,l) \geq \alpha \min 
\left( 1, \frac{f_j(E_j)}{f_{j'}(E_{j',j})} \right) 
\sum_{(i,l) \in {\cal E}_j \cap {\cal E}_{j'}} r_{j'}(i,l),
\]
where $E_{j',j} \subseteq E_{j'}$ denote 
the set of edges in $E_{j'}$ incident to buyers transacting with $j$, and
$\min (1, f_j(E_j)/f_{j'}(E_{j',j}))$ 
means a scaling factor 
adjusting the difference of supplies of $j,j'$ (viewed from $j$). 
For the instance in Section~\ref{subsec:example}, 
the midpoint clinching rule gives $0.896$-envy free sharing, whereas
the rules (i) and (ii) give $0.666$- and $0.727$-envy free sharing, respectively.
It fits our intuition that the midpoint rule is fair.
A further analysis of our mechanism 2
and developing fair clinching rules by means of the envy-freeness
deserve interesting future research.

Finally we mention some of remaining issues and future research directions.
	\begin{itemize}
		\item	Our mechanisms terminate 
		when all buyers have no demand. 
		The required number of iterations is obviously 
		bounded by $\sum_{i\in N}\lceil v'_i/\varepsilon \rceil$.
		In the computational complexity point of view, 
		our mechanisms are not  
		polynomial time algorithms, though each iteration 
		can be done in polynomial time (Theorem~\ref{greedy}).
		This is already the case for
		one-sided clinching-type mechanisms~\cite{DLN2012,GMP2015}.
		Acceleration of these mechanisms has not been studied so far.
		It would be interesting 
		to improve clinching-type mechanisms 
		to terminate in polynomial number of iterations.
		\item 	In this paper, we dealt with divisible goods. 
		We do not know whether our framework and results can be extended to auctions
		with indivisible goods. 
		Even if each $f_{j}\ (j\in M)$ is integer-valued, 
		$f_{w,d}$ is not necessarily integer-valued. 
		This prevents our mechanisms from dealing with indivisible goods.
		Notice that this is already the case for the original polyhedral clinching auction. 
		To extend our framework for indivisible goods, 
		the following work may be instrumental:
		Fiat et al. \cite{FLSS2011} presented 
		a mechanism enjoying (IC), (IR), and (PO) for one-sided markets, 
		where the seller has several kinds of goods and one indivisible unit for each goods. 
		Colini-Baldeschi et al. \cite{BHLS2015} extended 
		it to the multiple unit case.
		Goel et al. \cite{GLMNP2016} addressed two-sided markets with indivisible goods
		 and presented a clinching-type mechanism, though buyers are non-budgeted.

		\item 
		It would be interesting to study further generalizations 
		of our framework.
		For example, suppose that 
		each buyer has a preference on sellers.
		In this situation, our mechanism~2 
		can naturally incorporate the preference.
		Indeed, each buyer can clinch a maximal vector according to 
		the greedy algorithm using the ordering of his preference.
		Adding polymatroid constraints to both buyers and sellers is 
		also a natural generalization 
		to which the polymatroidal-flow approach could be applicable.
		Although extending polymatroid constraints to general polyhedral constraints	can cause an impossibility theorem (already for one-sided markets)~\cite{GMP2015},  
			several nice classes of 
			polyhedra generalizing polymatroids are known  
			in the literature of combinatorial optimization.
			{\em Bisubmodular polyhedra} are 
			such a generalization of polymatroids; see~\cite[Section 3.5]{F2005}.
			This polyhedron shares several nice common features with polymatroid, 
			such as greedy algorithm and several polynomial-time operations.
			It may be interesting to explore 
			the power of bisubmodular polyhedra 
			for the design of auctions and mechanisms.
	\end{itemize}

\section{Proofs}\label{sec:proof}
	In this section, we give 
	proofs of claimed results in Section~\ref{sec:mechanisms}.
	In the market $(N,M,E)$, 
	for node subset $Z$, let $E_Z$ denote the set of
	edges incident to nodes in $Z$.
	For edge subset $F$, let $N_F$ 
	denote the set of 
	buyers incident to $F$.

\subsection{Proof of Theorem \ref{greedy} and Proposition \ref{IRs}}\label{subsec:greedy}

	Let $d$ be the 
	demand vector, and $w$ the 
	transaction vector in line 4 in the current iteration.
	We start to study polytopes $P_{w,d}$, $P_{w,d}^i(\xi)$, and
	$P_{w,d}^i$ from the viewpoint of polymatroidal network flow.
	For each seller $j$, define polytope $P_{j,w} \subseteq \mathbb{R}_+^{E_j}$ by
	\begin{equation*}
	P_{j,w}:=\{x\in \mathbb R_{+}^{E_j}\mid w|_{E_j}+x\in P_j\}	.
	\end{equation*}
	Since $P_{j,w}$ is a
	contraction of $P_j$ (see; e.g. Fujishige \cite[Section 3.1]{F2005}), $P_{j,w}$ is a polymatroid, and the corresponding 
	submodular function $f_{j,w}:2^{E_j}\to\mathbb R_+$ 
	is given by
	\begin{equation*}
	f_{j,w}(F):=\min_{F'\supseteq F}\{f_j(F')-w(F')\}\quad (F\subseteq E_j).
	\end{equation*}
	This fact is also observed 
	from (i) $x(F') \leq f_j(F') - w(F')$ implies $x(F) \leq f_j(F') - w(F')$ for $F \subseteq F'$ by the nonnegativity of $x$, 
	(ii) $f_{j}(F') - w(F') \geq 0$ since $w$ is a feasible transaction, and
	(iii) if $f_{j,w}(F) = f_j(F')-w(F')$ and $f_{j,w}(G) = f_j(G') - w(G')$, 
	then 
	\begin{eqnarray}
	&& f_{j,w}(F) + f_{j,w}(G) \nonumber \\
	&& = f_j(F')- w(F') + f_j(G') - w(G') \nonumber \\
	&& \geq  f_j(F' \cap G') - w(F' \cap G') +  f_j(F' \cup G') - w(F' \cup G') \nonumber \\
	&& \geq f_{j,w}(F \cap G) + f_{j,w} (F \cup G). \label{eqn:FG}
	\end{eqnarray}
	We also define 
	$f, f_w:2^{E}\to \mathbb R_+$ by
	\begin{align*}
	f(F)&:=\sum_{j\in M}f_j (E_j\cap F) 
	\quad(F\subseteq E),\\
	f_w(F)&:=\min_{F'\supseteq F}\{f(F')-w(F')\}
	=\sum_{j\in M}f_{j,w}(E_j\cap F) \quad (F\subseteq E). 
	\end{align*}
	Note that both $f$ and $f_w$ are 
	monotone submodular. In particular, $f$  
	corresponds to the polymatroid $P$
	(by a theorem in McDiarmid \cite{MD1975}).

	Let us construct a polymatroidal network 
	from the market $(N,M,E)$, demand $d$, 
	and polymatroids $P_{j,w}$ for each $j\in M$;  see Figure \ref{network_2}.
	\begin{figure}[t]
		\centering
		\includegraphics[width=10cm]{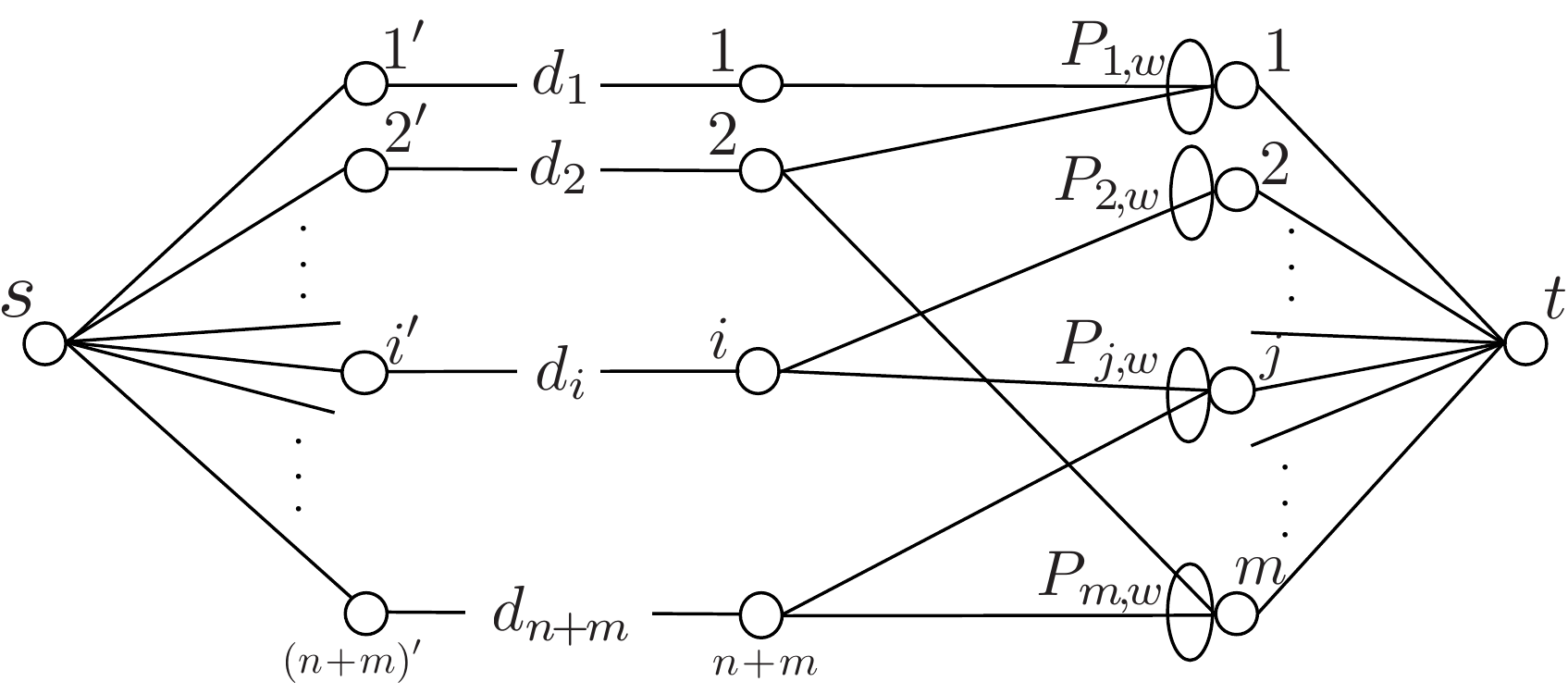}
		\caption{The polymatroidal network $\mathcal N$}
		\label{network_2}
	\end{figure}
	Each edge in $E$ is directed from $N$ to $M$.
	For each buyer $i$, consider copy $i'$ of $i$ and (directed) edge $i'i$.
	The set of the copies of $X\subseteq N$ is denoted by $X'$.
	Add source node $s$ and edge $si'$ for each $i'\in N'$.
	Also, add sink node $t$ and edge $jt$ for each $j\in M$.
	Edge-capacity $c$ is defined by 
	\begin{align}
	\label{capacity}
	c(e):=
	\begin{cases}
	d_i\quad \,{\rm if}\ e=i'i \ 
	{\rm for}\ i\in N, \\
	\infty\quad {\rm otherwise}.	
	\end{cases}
	\end{align}
	Two polymatroids of each node  are defined as follows.
	For $j\in M$, the polymatroid on $\delta^{-}_{j}=E_j$ is
	defined as $P_{j,w}$, and 
	the other polymatroid on $\delta^+_{j}$ is 
	defined by the capacity (\ref{capacity}) as in (\ref{eqn:capacity}).
	Two polymatroids of other nodes are defined by
	the capacity. 
	Let $\mathcal N$ denote the resulting 
	polymatroidal network.

	Now the remnant supply polytope $P_{w,d}$ 
	is written by using flows in $\mathcal N$:
	\begin{equation*}
	P_{w,d}=\{x\in \mathbb R_+^{E}\mid
	\exists \varphi:{\rm flow\ in}\ \mathcal N, x=\varphi|_E\}.
	\end{equation*}
	In the proof, the following function 
	$f_{w,d}:2^{E}\to \mathbb R_+$ plays a key role
	\begin{equation}
	\label{max-flow}
	 f_{w,d}(F):=\max_{x\in P_{w, d}} x(F)
	 =\max_{\varphi:\,
	 {\rm flow}\, {\rm in}\, \mathcal N}\varphi(F)\quad (F\subseteq E).
	\end{equation}
	Then we obtain the following:
	\begin{theorem}
	\label{property_f_wd}
	For $F\subseteq E$, it holds 
	\begin{equation}
	\label{f_wd}
	 f_{w,d}(F)=\min_{X\subseteq N_F}\{ f_w(E_X\cap F)+d(N_F\setminus X)\},
	\end{equation}
	and $f_{w,d}(F)$ can be computed 
	in polynomial time,
	provided the value oracle of each $f_j$ is available.
	\end{theorem}
	\begin{proof}

	We modify $\mathcal N$ so that $f_{w,d}(F)$ is 
	equal to the maximum flow value. Change the 
	capacity of 
	each $e\in E\setminus F$ to $0$. Also, for each 
	$i\in N\setminus N_F$, change $c(i'i)$ to $0$.
	Then the maximum flow value is equal to $f_{w,d}(F)$.
	By Theorem \ref{max-flow-min-cut}, the value 
	$f_{w,d}(F)$ is equal to 
	(\ref{max-flow-min-cut2}).
	Observe that subset $U$ attaining the minimum can take a form 
	of $(N'\cup X)+s$ for $X\subseteq N_F$.
	Then the set~$\delta^+_U$ of edges leaving $U$ is equal to   
	$\{i'i \mid i\in N_F\setminus X\} \cup E_X$.
	In a partition $\{A,B\}$ of $\delta^+_U$ 
	attaining minimum, 
	edge $i'i$ contributes $d_i$ in 
	both cases of $i'i\in A$ 
	and $i'i\in B$.	
	Each edge $e$ in $E_X \cap F$ must be in $A$ (by $c(e)=\infty$), and each edge  $e$
	in $E_X\setminus F$ can be in $B$
	(by $c(e)=0$ and monotonicity of $f_w$). Thus we have~(\ref{f_wd}).

	 In the modified network ${\cal N}$, $f_{w,d}(F)$ 
	can be computed by solving a maximum polymatroidal flow problem.
	There are several polynomial time algorithms for this problem; 
	see Fujishige \cite[Section 5.5]{F2005}.
	\end{proof}

	Let $\xi \in \mathbb R^{E_i}_+$.
	We consider the remnant supply polytope $P^{i}_{w,d}(\xi)$ of remaining buyers except $i$
	and the clinching polytope $P^{i}_{w,d}$ of $i$; see (\ref{remnant_supply_i}) and (\ref{eqn:clinching_polytope}) for definition.
	Define $h^i_{\xi}:2^{N-i} \to \mathbb R_{+}$  and $h_{w,d}^{i}:2^{E_i} \to \mathbb {R}_+$ by
	\begin{align}
	h^{i}_{\xi}(Y)& :=\min_{F\subseteq E_i}\{ f_{w,d}(E_Y\cup F)-
		\xi(F)\}\quad (Y\subseteq N-i), \label{eqn:h^i_xi} \\
		h_{w,d}^{i}(F)& := f_{w,d}(E_{N-i}\cup F)-
		f_{w,d}(E_{N-i})\quad (F\subseteq E_i). \nonumber
	\end{align}
	Our goal is to show the following:
	\begin{theorem}
	\label{h^i_xi}
	\begin{itemize}
	\item[{\rm (i)}] $P^{i}_{w,d}(\xi)$ is nonempty if and only if 
	$h^{i}_{\xi}$ is nonnegative valued; 
	 in this case,  
	$h^{i}_{\xi}$ is monotone submodular on $N-i$, and 
	$P^{i}_{w,d}(\xi)$ coincides with 
	the polymatroid associated with $h^i_{\xi}$.
	\item[{\rm (ii)}] $h_{w,d}^{i}$ is monotone submodular on $E_i$, and 
	$P^i_{w,d}$ coincides with 
	the polymatroid associated with~$h_{w,d}^{i}$.
	\end{itemize}
	\end{theorem}
	 Theorem \ref{greedy} is an immediate consequence
	 of Theorems \ref{property_f_wd} and \ref{h^i_xi}~(ii).
	\begin{proof} 
		(i).
	We utilize the polymatroidal network $\mathcal N$.
	Then $u\in P_{w,d}^{i}(\xi)$ if and only if 
	there exists a flow $\varphi$ in $\mathcal N$ 
	such that $\varphi(ij)=\xi_{ij}$
	for $ij\in E_i$, $\varphi(i'i)=\xi(E_i) \leq d_i$, and
	$\varphi(sk')=u_k$ for $k\in N-i$.
	By modifying $\mathcal N$, 
	we give an equivalent condition as follows.
	Replace $ij\in E_i$ by $sj$ with capacity $\xi_{ij}$,
	change capacity $c(sk')$ of $sk'$ to $u_k$ for $k\in N-i$,
	and delete nodes $i'$ and $i$. 
	By $\xi(E_i)\leq d_i$ and Theorem~\ref{max-flow-min-cut}, the above condition  for $u\in P_{w,d}^{i}(\xi)$ can be  
	written as:
	\begin{itemize}
	\item The maximum flow value is equal to 
	$\xi(E_i)+u(N-i)$, or 
	equivalently, $U=\{s\}$ and $(A,B)=(\emptyset, \delta^{+}_s)$ attain the minimum of
	(\ref{max-flow-min-cut2}).
\end{itemize}
	In the minimum of (\ref{max-flow-min-cut2}), 
	it suffices to consider $U$ of form $U=(Y'\cup Z)+s$ 
	for $Z\subseteq Y\subseteq N-i$. 
	(If $k \in U \not \ni k'$, then we can remove $k$ from $U$ 
	without increasing the cut value.)
	Then the set $\delta^+_U$ of edges leaving $U$ 
	is equal to   
	$E_i\cup \{sk'\mid k\in N\setminus (Y+i)\}\cup \{k'k \mid k\in Y\setminus Z\} \cup E_Z$.
	In any partition $\{A,B\}$ of $\delta^+_U$,
	edge $sk'$ contributes $u_k$, and
	$k'k$ contributes $d_k$.
	It must hold $E_Z\subseteq A$.
	Thus it suffices to consider 
	a partition $\{F,E_i\setminus F\}$ of $E_i$.	
	Then $u\in P^i_{w,d}(\xi)$ if and only if
	\begin{align*}
	\xi(E_i)+u(N-i) \leq 
	f_w(E_Z\cup F)+\xi(E_i\setminus F)
	+u((N-i)\setminus Y)+d(Y\setminus Z)
	\end{align*}
	for $F\subseteq E_i$ and $Z\subseteq Y\subseteq N-i$.
	This is further equivalent to 
	\begin{equation}
	\label{u_y}
	 u(Y)\leq\min_{Z\subseteq Y,F\subseteq E_i}\{f_w(E_Z\cup F)+d(Y\setminus Z)-\xi(F)\}\quad (Y\subseteq N-i).
	\end{equation}
	In particular, $P^i_{w,d}(\xi)$ is nonempty if and only if  the right hand side of (\ref{u_y}) is nonnegative.
	Define $\tilde{h}^i_{\xi}:2^{N-i}\to\mathbb R_+$ by the right hand side of
	(\ref{u_y}). 
	It suffices to show that $\tilde{h}^i_{\xi}$ is equal to $h^i_\xi$, and 
	is monotone submodular if it is nonnegative valued.
	
	We first show $h^i_{\xi}(Y)=\tilde{h}^i_{\xi}(Y)$ for $Y \subseteq N- i$.
	By substituting the formula (\ref{f_wd}) 
	of $f_{w,d}$ for the definition~(\ref{eqn:h^i_xi}) of $h^{i}_\xi$,  
	we obtain
	\begin{align}
	\label{hixi}
	h^{i}_{\xi}(Y)&=\min_{F\subseteq E_i}\{\min_{Z'\subseteq N_{E_Y \cup F}}\{f_w(E_{Z'}\cap (E_Y\cup F))+d(N_{E_Y\cup F}\setminus Z')\}-\xi(F)\}.
	\end{align}
	Observe that 
	\begin{equation*}
	N_{E_Y \cup F} = 
	\left\{ 
	\begin{array}{cc}
	Y + i & {\rm if}\ F \neq \emptyset, \\
	Y  & {\rm if}\ F = \emptyset.   
	\end{array}
	\right.
	\end{equation*}
	Therefore the minimum in (\ref{hixi}) is 
	taken over $(F,Z')$ with $\emptyset \neq F \subseteq E_i$, $Z' \subseteq Y +i$ and over $(F,Z')$ with $F= \emptyset$, $Z' \subseteq Y$. 
	Notice that if $F=\emptyset$, then $Z'$ and $Z'-i$
	have the same value in~(\ref{hixi}).
	Thus $h^{i}_{\xi}$ is also written as
	\begin{align}\label{hixi'}
	h^{i}_{\xi}(Y) &=\min_{Z'\subseteq Y+i,\ F\subseteq E_i}
	\{f_w(E_{Z'} \cap  (E_Y \cup F))+d(N_{E_Y\cup F}\setminus Z')-\xi(F)\}.
	\end{align}
	We consider the minimum of (\ref{hixi'}) over all $Z'\subseteq Y$ 
	and over all $Z'=Z+i$ with $Z\subseteq Y$.
	The first minimum $\alpha$ is equal to 
	\begin{align*}
	&\min_{Z'\subseteq Y,\ F\subseteq E_i}
	\{f_w(E_{Z'})+d(N_{E_Y\cup F}\setminus Z')-\xi(F)\}\\
	&=\min_{Z'\subseteq Y,\ F\subseteq E_i}
	\{f_w(E_{Z'})+d(Y\setminus Z')+d(N_F)-\xi(F)\}\\
	&=\min_{Z'\subseteq Y}
	\{f_w(E_{Z'})+d(Y\setminus Z')\},
	\end{align*}
	where $F=\emptyset$ attains the minimum since $d_i\geq \xi(F)$.
	The second minimum $\beta$ is  equal to 
	\begin{align*}
	&\min_{Z\subseteq Y,\ F\subseteq E_i}
	\{f_w(E_{Z+i}\cap (E_Y\cup F))+d(N_{E_Y\cup F}\setminus (Z+i))-\xi(F)\}\\
	&=\min_{Z\subseteq Y,\ F\subseteq E_i}
	\{f_w(E_Z\cup F)+d(Y\setminus Z)-\xi(F)\}\\
	&=\tilde{h}^i_{\xi}(Y)
	\end{align*}
	Observe that $\alpha\geq \beta$ 
	(since the first is the special case of $F = \emptyset$ in the second). 
	Thus $h^i_{\xi}(Y)=\min\{\alpha,\beta\}
	=\beta=\tilde{h}^i_{\xi}(Y)$ 
	as required.

	Suppose that $\tilde{h}^i_{\xi}$ is nonnegative valued. 
	Here $\tilde{h}^i_{\xi}(Y)$ is interpreted as 
	the maximum value of a flow from $Y'$ to $t$ 
	in ${\cal N}$ with the transaction $w$ replaced by $w + \xi$, 
	where $\xi$ is extended as 
	$\xi: E \to \mathbb {R}_+$ by $\xi(e) := 0$ for $e \in E \setminus E_i$.
	Note that $w + \xi$ is actually a feasible transaction
	by $f(F) - (w + \xi)(F) \geq f_w(F) - \xi (F) \geq 0$, where the last inequality follows from the nonnegativity of (\ref{u_y}) (with $Y = \emptyset$).
	Via Theorem~\ref{property_f_wd}, 
	the maximum flow value is equal to $\min_{Z\subseteq Y}\{f_{w+\xi} (E_Z)+d(Y\setminus Z)\}$.  Then, substitute 
	$f_{w+\xi}(E_Z) = \min_{H \supseteq E_Z} \{ f(H) - w(H) - \xi (H)\} =  \min_{H \supseteq H' \supseteq E_Z} \{ f(H) - w(H) - \xi (H')\} =
	\min_{H' \supseteq E_Z} \{ f_w(H') - \xi(H')\} =
	\min_{F \subseteq E_i}\{ f_{w}( E_Z \cup F) - \xi(F) \}$, 
	where we use $\xi(H') = \xi(E_i \cap H')$ and the monotonicity of $f_w$.
	Namely $\tilde{h}^i_{\xi}$ is viewed as a network induction 
	 of $f_w$ \cite{MD1975}, and is necessarily monotone submodular;
	one can also see the submodularity directly by taking
	minimizers of $\tilde{h}^i_{\xi}(X)$ and $\tilde{h}^i_{\xi}(Y)$ (as in (\ref{eqn:FG}))

	(ii). By (i), we have 
	$P^{i}_{w,d}(\xi)=P^{i}_{w,d}(0)$ if and only if 
	$h^{i}_{\xi}(X)=h^{i}_{0}(X)$ 
	for each $X\subseteq N-i$.
	By the monotonicity of $ f_{w,d}$, the latter condition is 
	equivalent to 
	\begin{equation*}
	(h_{\xi}^i(X) :=)\min_{F\subseteq E_i}\{ f_{w,d}(E_X\cup F)-\xi(F)\} = f_{w,d}(E_X) (= h_0^i(X) \geq 0)
	\quad (X\subseteq N-i).
	\end{equation*}
	Notice that this condition also guarantees $P_{w,d}^i(\xi) \neq \emptyset$.
	Thus  $\xi\in \mathbb R_+^{E_i}$ belongs to $P^i_{w,d}$ 
	if and only if it holds 
	\begin{align}
	 \xi(F)&\leq f_{w,d}(E_X\cup F)- f_{w,d}(E_X) \label{clinch}
	 \end{align}
	 for each $F \subseteq E_i$ and $X \subseteq N-i$.
	 The latter condition is equivalent to
	 \begin{align} 
	\xi(F) &\leq \min_{X \subseteq N-i}
	 \{ f_{w,d}(E_X\cup F)- f_{w,d}(E_X)\} \nonumber \\
	 &=\min_{X \subseteq N-i}
	 \{ f_{w,d}(E_{X+i}\setminus (E_i\setminus F))
	 - f_{w,d}(E_X\setminus (E_i\setminus F))\} \quad (X\subseteq N-i), \label{clinch0} 
	 \end{align}
	 where we use $E_X \cup F = E_{X+i} \setminus (E_i \setminus F)$ and
	 $E_X = E_X \setminus (E_i \setminus F)$ for $X \subseteq N-i$.
	 By Lemma~\ref{submodularity}~(i) shown below, 
	 the function 
	 	\begin{equation*}
	 	X \mapsto f_{w,d}(E_X\setminus (E_i\setminus F))=\\
	 	\begin{cases}
	 	f_{w,d}(E_{X-i}\cup F)\quad \ {\rm if}\ i\in X, \\
	 	f_{w,d}(E_X)         \qquad\qquad  {\rm otherwise}
	 	\end{cases}
	 	\end{equation*}
	 is monotone submodular on $N$. 
	 By (\ref{submo0}), 
	 the minimum in (\ref{clinch0}) is attained by $X = N-i$.
	 Thus we have 
	 \begin{align}
	 \xi(F) &\leq f_{w,d}(E_{N}\setminus (E_i\setminus F))-
	 f_{w,d}(E_{N-i}\setminus (E_i\setminus F)) \nonumber\\
	&= f_{w,d}(E_{N-i}\cup F)-
	 f_{w,d}(E_{N-i}) = h_{w,d}^{i}(F) \nonumber,
	\end{align}
    The monotone submodularity of $F \mapsto f_{w,d}(E_{N-i} \cup F)$ is shown in the next lemma (Lemma \ref{submodularity}~(ii)).
	This concludes that the clinching polytope $P^i_{w,d}$
	is the polymatroid associated with $h_{w,d}^{i}$, as required.
	\end{proof}

	\begin{lemma}
		\label{submodularity}
		\begin{itemize}
			\item[{\rm (i)}] For $F \subseteq E_i$, the function 
			$X \mapsto f_{w,d}(E_X\setminus (E_i\setminus F))$
			is monotone submodular on $N$.
			\item[{\rm (ii)}] For $X \subseteq N-i$, 
			the function $G\mapsto f_{w,d}(E_X\cup G) - f_{w,d}(E_X)$
			is monotone submodular on~$E_i$.	
		\end{itemize}
	\end{lemma}
		Note that $f_{w,d}$ is not necessarily submodular on $E$.
	\begin{proof}
		The monotonicity of both functions in (i) and (ii) is 
		immediate from the monotonicity of~$f_{w,d}$ in~(\ref{max-flow}).
		We show the submodularity.
		
		(i) Let $G := E_i \setminus F$. 
		The submodularity of $X \mapsto f_{w,d}(E_X \setminus G)$ can directly be
		shown by using formula
		(\ref{f_wd}) and taking minimizers of 
		$f_{w,d}(E_X\setminus G)$ and 
		$f_{w,d}(E_Y\setminus G)$ (as in (\ref{eqn:FG})).
		Instead,  we consider a flow interpretation. Consider $\mathcal N$,
		and remove all edges in $G$.
		Then $f_{w,d}(E_X\setminus G)$ is equal to
		the maximum value of a flow from $X'$ to $t$ in the network.
		Thus $X\to f_{w,d}(E_X\setminus G)$
		is a network induction 
		of $f_{w}$, 
		and is submodular. 
		
		(ii) Let $F,G\subseteq E_i$. 
		It suffices to consider the case 
		where $F\not\subseteq G$ and $G\not\subseteq F$.
		In particular, both $F$ and $G$ are nonempty.
		Hence $N_{E_X\cup F}=N_{E_X\cup G}=X+i$.
		Consider minimizers $Y$ and $Z$, respectively, of 
		$f_{w,d}(E_X\cup F)$ and
		$f_{w,d}(E_X\cup G)$ in (\ref{f_wd}).
		Namely $f_{w,d}(E_X\cup F)=f_w(E_Y\cap (E_X\cup F))+d((X+i)\setminus Y)$ and $f_{w,d}(E_X\cup G)=f_w(E_Z\cap (E_X\cup G))+d((X+i)\setminus Z)$.
		Here $Y, Z \subseteq X + i$.
		\\
		Case 1: Both $Y$ and $Z$ contain $i$. 
		Notice that $E_Y \cap (E_X \cup F) = E_{Y - i} \cup F$ and the same holds for $Z$.
		Then we have
		\begin{align*}
		&f_{w,d}(E_X\cup F)+ f_{w,d}(E_X\cup G)\\
		&\quad=f_w(E_{Y-i}\cup F)
		+d((X+i)\setminus Y)
		+ f_w(E_{Z-i}\cup G)+
		d((X+i)\setminus Z)\\
		&\quad\geq  f_w((E_{Y-i}\cup E_{Z
			-i})\cup (F\cup G))
		+d((X+i)\setminus (Y\cup Z))\\
		&\qquad+ f_w((E_{Y-i}
		\cap E_{Z-i})\cup(F\cap G))
		+d((X+i) \setminus (Y\cap Z)) \\
		&\quad=  f_w(E_{(Y\cup Z)-i}\cup (F\cup G))
		+d((X+i)\setminus (Y\cup Z)) \\
		&\qquad + f_w(E_{(Y\cap 
			Z)-i}\cup (F\cap G))
		+d((X+i) \setminus (Y\cap Z)) \\
		&\quad\geq  f_{w,d}(E_X\cup (F\cup G))+ f_{w,d}(E_X\cup (F\cap G)).
		\end{align*}
		Case 2: At least one of $Y$ and 
		$Z$ does not contain $i$, say $i \not \in Z$. 
		By $X+i = N_{E_X \cup G}$, we have
		\begin{align*}
		f_{w,d}(E_X\cup G)&=\min_{Z\subseteq X}
		\{f_w(E_Z\cap (E_X\cup G))+d((X+i)\setminus Z)\}\\
		&=\min_{Z\subseteq X}
		\{f_w(E_Z)+d(X\setminus Z)+d_i\}
		=f_{w,d}(E_X)+d_i.
		\end{align*}
		Then we obtain
		\begin{align*}
		f_{w,d}(E_X\cup F)+f_{w,d}(E_X\cup G)&=
		f_{w,d}(E_X\cup F)+f_{w,d}(E_X)+d_i\\
		&\geq f_{w,d}(E_X\cup (F\cap G))+f_{w,d}(E_X)+d_i\\
		&\geq f_{w,d}(E_X\cup (F\cap G))+f_{w,d}(E_X\cup (F\cup G)),
		\end{align*}
		where the first inequality follows from the monotonicity of $f_{w,d}$ and 
		the second inequality follows from $f_{w,d}(E_X\cup (F\cup G)) \leq f_{w,d}(E_X)+d_i$, i.e., the total flow on $E_X \cup (F \cup G)$
		is at most the sum of the total flow on $E_X$ and capacity $d_i$ of $i'i$.
		This concludes that $G\mapsto f_{w,d}(E_X\cup G)$ is submodular.
	\end{proof}

	By (\ref{clinch}) in the proof of Theorem~\ref{h^i_xi}~(i) 
	and Theorem~\ref{h^i_xi}~(ii), 
	we obtain the following 
	corollary.
	\begin{corollary}
	\label{clinch_limit}
	Suppose that buyer $i$ clinches 
	$\xi\in \mathbb R^{E_i}_+$ of goods.
	\begin{itemize}
		\item[{\rm (i)}] For $F\subseteq E_i$ and $X \subseteq N-i$, 
		it holds $\xi(F)\leq f_{w,d}(E_X\cup F)- f_{w,d}(E_X)$. 
	\item[{\rm (ii)}] 	Buyer $i$ gets $\xi(E_i)=f_{w,d}(E)- f_{w,d}(E\setminus E_i)$ 
	amount of goods.
    \end{itemize}
\end{corollary}

Next we prove Proposition~\ref{IRs} $=$ (IRs), which is an immediate consequence 
of the following.
	\begin{lemma}
	\label{transaction}
	For each seller $j$, 
	if $c_{n+j}<\rho_j$, then $w_{ij}=0$ for all $ij \in E_j - (n+j)j$.
	\end{lemma}
	This means that the goods of each buyer $j$ are sold by nonvirtual buyer 
	at price at least $\rho_j$. 
	Therefore, after transforming the allocation 
	of Algorithm \ref{algo2}
	into the allocation without virtual buyers, it holds 
	$r_j \geq \rho_j w(E_j)$.
	Then we obtain (IRs):
	$u_j(\mathcal M(\mathcal I))= 
	r_j+\rho_j(f_j(E_j)-w(E_j))\geq\rho_j f_j(E_j)$.
	\begin{proof}
		It suffices to show that 
		if $c_{n+j} < \rho_j$ and $w|_{E_j - (n+j)j} = 0$ then  
		$\xi|_{E_j - (n+j)j} =0$; consequently 
		$w|_{E_j - (n+j)j} = 0$ until $c_{n+j} = \rho_j$.
		By Corollary \ref{clinch_limit}~(i) for $X =\{n+j\}$ and $F=\{ij\}$ with $i \neq n+j$, 
		the clinching vector 
		$\xi\in\mathbb R^{E_i}_+$ 
		of buyer $i$ satisfies  
		\begin{align*}
		\xi_{ij}\leq 
		f_{w,d}(\{(n+j)j\}\cup \{ij\})-
		f_{w,d}(\{(n+j)j\}).
		\end{align*}
		We show that the right hand side is zero.
		Now $d_{n+j}=\infty$ by $c_{n+j} < \rho_j$.
		Hence in the formula~(\ref{f_wd}) of $f_{w,d}(\{(n+j)j\}\cup \{ij\})$, 
		the minimum is attained by $X = \{ n+j\}$ or $\{i,n+j\}$.
		Thus we have $f_{w,d}(\{(n+j)j\}\cup \{ij\}) = \min\{ f_{w}(\{(n+j)j\}\cup \{ij\}),
		f_{w}(\{(n+j)j\})+d_{i}\}$. Moreover we have
		\begin{align*}
		f_{w}(\{(n+j)j\}\cup \{ij\}) & = f_{w,j}(\{(n+j)j\}\cup \{ij\}) \\ 
		& = \min_{F \subseteq E_{j}: ij,(n+j)j \in F} f_j(F) - w(F)  \\
		& = f_j(E_j) - w_{(n+j)j}, 
		\end{align*}
		where we use $w_{ij} =0$ for $ij \in E_{j} - (n+j)j$ and the modification (\ref{f_j}) of $f_j$.
		Similarly $f_{w}(\{(n+j)j\}) =  f_j(E_j)- w_{(n+j)j}$, 
		and hence $f_{w,d}(\{(n+j)j\}\cup \{ij\}) = f_j(E_j)- w_{(n+j)j}$.
		Also $f_{w,d}(\{(n+j)j\}) = f_{w}(\{(n+j)j\}) =  f_j(E_j)- w_{(n+j)j}$.
		Thus we obtain $\xi_{ij}=0$.
	\end{proof}

\subsection{Proof of Theorem \ref{relation}}\label{subsec:relation}
	Here we give a proof of Theorem \ref{relation}. 
	Consider the transaction vector $w\in \mathbb R_+^{E}$　
	and demand vector $d\in \mathbb R_+^{N}$ in an iteration of Algorithm 2.
	Then feasible transaction $y:N \to \mathbb{R}_{+}$ 
	for the reduced market is naturally given by $y_i := w(E_i)$ $(i \in N)$.
	We consider the behavior of Algorithm 1 for this $y,d$.
	It is known in \cite{GMP2015} that 
	$\tilde{P}$ is a polymatroid, and 
	the corresponding submodular function $g$ is
	\begin{equation*}
	g(X)=f(E_X)\ \quad (X\subseteq N).
	\end{equation*}	
	Also $\tilde{P}_{y,d}$ is a 
	polymatroid, and the corresponding 
	submodular function 
	$g_{w,d}$ is given by 
	\begin{align*}
	g_{w,d}(X)&=\min_{Z\subseteq X}\{\min_{Z'\supseteq Z}\{g(Z')-w(E_{Z'})\}+d(X\setminus Z)\}.
	\end{align*}
	These facts are also obtained by considering the polymatroidal network 
	with one seller (of polymatroid constraint $\tilde P$), as in Figure~\ref{network_2}.  
	The clinch $\zeta_i$ in Algorithm \ref{algo1} is given as follows:
	\begin{theorem}[Goel et al. \cite{GMP2015}]
	\label{clinch_goel}
	In Algorithm 1, it holds 
	$\zeta_{i}=g_{w,d}(N)- g_{w,d}(N-i)$ 
	for each $i\in N$.
	\end{theorem}
	The goal of this section is to prove following:
	\begin{theorem}
	\label{f=g}
	In Algorithm 2, we have the following:
	\begin{itemize}
		\item [{\rm (i)}] $f_{w,d}(E_X)= g_{w,d}(X)$ holds for each $X\subseteq N$.
		\item[{\rm (ii)}]  The value $f_{w,d}(E)-f_{w,d}(E\setminus E_k)$ does not 
		change after the clinch of buyer $i\in N-k$.
	\end{itemize}
	\end{theorem}

	In particular, the total amount $\xi(E_i)$ 
	of the clinch of each buyer $i$ is  
	determined
	by $w$ and $d$ at the beginning of For Loop in line 4 	
	(i.e., independent of the ordering of buyers).
	Now $\epsilon$ is the same in both algorithms.
	By Corollary~\ref{clinch_limit}~(ii) and Theorems~\ref{clinch_goel} and 
	\ref{f=g}, we obtain Theorem~\ref{relation} that in Algorithm \ref{algo2} 
	each buyer obtains the same total amount of 
	goods as in Algorithm \ref{algo1}.
	
	The rest of this subsection is devoted to proving 
	Theorem \ref{f=g}.
	Define $\tilde{g}_{w,d}:2^N\to \mathbb R_+$ by
	\begin{equation}\label{eqn:tilde_g_wd}
		\tilde{g}_{w,d}(X):= f_{w,d}(E_X)=\min_{Z\subseteq X}\{f_w(E_Z)+d(X\setminus Z)\}\quad (X\subseteq N),
	\end{equation}
	where we use (\ref{f_wd}) for the second equality.
	We first analyze the behavior of $\tilde{g}_{w,d}$ when buyer $i$ clinches.
	Suppose that  $w$ and $d$ are the transaction  vector 
	and demand vector, respectively, at line 4 in the current iteration. 
	\begin{lemma}\label{lem:change}
	Suppose that buyer $i$ clinches 
	$\xi\in \mathbb R^{E_i}_+$ of goods, and 
	$w$ and $d$ are updated to $w'$ and $d'$, respectively, in lines 6--8. 
	Then, for $X \subseteq N$, the following hold:
	\begin{itemize}
		\item[{\rm (i)}] If $i \in X$, then $\tilde g_{w',d'}(X) = \tilde g_{w,d}(X) - \xi(E_i)$.
		\item[{\rm (ii)}] If $i \not \in X$, then $\tilde g_{w',d'}(X) = \tilde g_{w,d}(X)$.
	\end{itemize}
\end{lemma}
Now Theorem~\ref{f=g} (ii) is an immediate corollary of Lemma~\ref{lem:change}~(i) since
\begin{align}
f_{w',d'}(E)- f_{w',d'}(E \setminus E_k) & = \tilde g_{w',d'}(N)- \tilde g_{w',d'}(N-k) \nonumber \\
& =  (\tilde g_{w,d}(N) - \xi(E_i)) - (\tilde g_{w,d}(N-k) - \xi(E_i)) \nonumber \\
 & =  f_{w,d}(E)- f_{w,d}(E \setminus E_k). \label{eqn:immediate}
\end{align}
We prove (i) in the end, and we here only prove (ii). 
	\begin{proof}[Proof of Lemma~\ref{lem:change}~(ii)]	
	Here 
	$w'(e) = w(e)$ for $e \in E\setminus E_i$  
	and $w'(e) = w(e)  + \xi(e)$ for $e \in E_i$, and  $d'_k = d_k$ for $k \neq i$. 
	Then
	$\tilde{g}_{w',d'}(X)$ is given by 
	\begin{align}
		\label{after_clinch}
		\tilde{g}_{w',d'}(X) = 
		\min_{Z\subseteq X}\{ f_{w'}(E_Z)+d(X\setminus Z)\}.
	\end{align}
	For $Z \subseteq X$, it holds
	\[
	f_{w'}(E_{Z}) =\min_{G\subseteq E_i}\{f_w(E_{Z}\cup G)-\xi(G)\}
	\]
	since $f_{w'}(E_{Z}) =\min_{H\supseteq E_{Z}}
	\{f(H)-w'(H)\} = \min_{H \supseteq H' \supseteq E_{Z}}
	\{f(H)-w(H) - \xi(H' \cap E_i)\} = \min_{H' \supseteq E_{Z}} \{ f_w(H') - \xi(H' \cap E_i)\} 
	= \min_{G \supseteq E_i} \{ f_w(E_{Z} \cup G) - \xi(G)\}$; 
	in the last equality put $G = H' \cap E_i$ and use the monotonicity 
	$f_w(H') \geq f_w(E_Z \cup G)$. 
	Then the inequality ($\leq$) follows from  
	\begin{align*}
		\tilde{g}_{w',d'}(X)&=\min_{Z\subseteq X}
		\{\min_{G\subseteq E_i}
		\{f_{w}(E_{Z}\cup G)-\xi(G)\}+d(X\setminus Z)\}\\
		&\leq\min_{Z\subseteq X}
		\{f_{w}(E_{Z})+d(X\setminus Z)\}=f_{w,d}(E_X)=\tilde{g}_{w,d}(X).
	\end{align*}
	($\geq$) follows from
	\begin{align*}
		\tilde{g}_{w',d'}(X)&=\min_{Z\subseteq X}
		\{\min_{G\subseteq E_i}
		\{f_{w}(E_{Z}\cup G)-\xi(G)\}+d(X\setminus {Z})\}\\
		&=\min_{G\subseteq E_i}\{\min_{Z \subseteq X}
		\{f_{w}(E_{Z}\cup G)+d(X\setminus Z)\}-\xi(G)\}\\
		&\geq \min_{G\subseteq E_i}
		\{f_{w,d}(E_X\cup G)-\xi(G)\} \\
		&=f_{w,d}(E_X)=\tilde{g}_{w,d}(X),
	\end{align*}
	where the third equality follows from Corollary \ref{clinch_limit}~(i), 
	and the inequality can be verified from the definition (\ref{max-flow}) of $f_{w,d}(E_X\cup G)$. 
	Namely the total flow on 
	$E_X \cup G$ is at most 
	the sum of the total flows on 
	$E_Z \cup G$ and on $\{ i'i \mid i \in X \setminus Z\}$:
	the former is bounded by $f_w(E_Z \cup G)$ and 
	the latter is bounded by $d(X \setminus Z)$.
\end{proof}	
	We next show Theorem~\ref{f=g}~(i), i.e., 
	$g_{w,d}(X) = \tilde g_{w,d}(X) (:= f_{w,d}(E_X)$).
	Observe from the definitions that the two functions are written as:
	\begin{align*}
	g_{w,d}(X) & =\min_{Z\subseteq X}\{\min_{Z'\supseteq Z}\{f(E_{Z'})-w(E_{Z'})\}+d(X\setminus Z)\}. \\
	\tilde g_{w,d}(X) & = \min_{Z \subseteq X}\{\min_{F\supseteq E_Z} \{f(F)-w(F)\}+d(X\setminus Z)\}.
	\end{align*} 
	Comparing the inner minimums, we see that $g_{w,d}(X) \geq \tilde g_{w,d}(X)$ generally holds. 
	We are going to show the converse.
	By a {\em minimizer} of $\tilde{g}_{w,d}(X)$ 
	we mean a subset $Z \subseteq X$ satisfying $\tilde{g}_{w,d}(X) = \min_{F\supseteq E_Z} \{f(F)-w(F)\} + d(X\setminus Z)$. 
	Our goal is to show that in Algorithm 2 there always exists a minimizer $Z$ of $\tilde g_{w,d}(X)$ 
	such that $f(E_Z) - w(E_Z) = \min_{F\supseteq E_Z} \{f(F)-w(F)\} 
	( = f_w(E_Z))$, which implies $g_{w,d}(X) = \tilde g_{w,d}(X)$.

	\begin{lemma}
	\label{minimal_minimizer}
	For $X\subseteq N$, the following conditions are equivalent.
	\begin{itemize}
	\item[{\rm (i)}] $X$ is a (unique) minimal minimizer 
	of $\tilde{g}_{w,d}(X)$; in this case, $\tilde{g}_{w,d}(X) = f_w(E_X)$ holds.
	\item[{\rm (ii)}] $X$ is a minimal minimizer of $\tilde{g}_{w,d}(Y)$ 
		  for some $Y$ with 
		  $X\subseteq Y\subseteq N$.
	\end{itemize}
	\end{lemma}	
	\begin{proof}
	It suffices to show that (ii) implies (i).
	Suppose that 
	$X$ is a minimal minimizer of $\tilde{g}_{w,d}(Y)$
	but $X$ is not a minimal minimizer of $\tilde{g}_{w,d}(X)$.
	Then there exists a proper subset 
	$X'\subset X$ such that
	$
	 f_w(E_X) \geq f_w(E_{X'})+d(X\setminus X')
	$.
	Then
	$
	 f_w(E_X)+d(Y\setminus X) \geq  f_w(E_{X'})+d(X\setminus X')+d(Y\setminus X)
	= f_w(E_{X'})+d(Y\setminus X')$,
	which contradicts the minimality of $X$ for $\tilde g_{w,d}(Y)$.
	\end{proof}
	Let $\mathcal U \subseteq 2^N$ denote the family of subsets $X\subseteq N$ satisfying the conditions in Lemma~\ref{minimal_minimizer}.
	Namely, $X \in {\cal N}$ is a minimal minimizer of $\tilde g_{w,d}(Y)$ 
	for some $Y$.
	By the argument before Lemma~\ref{minimal_minimizer}, the following (ii) completes the proof of Theorem~\ref{f=g}~(i). 
	\begin{proposition}
	\label{min}
	In Algorithm 2, 
	the following hold:
	\begin{itemize}
		\item[{\rm (i)}] The family $\mathcal U$ is equal to $2^N$ at the beginning, and is monotone nonincreasing through the iterations.
		\item[{\rm (ii)}] For each  $X \in \mathcal U$, 
		it always holds
		\begin{equation}
		\label{tilde_g}
		f_w (E_X) = f(E_X)-w(E_X).
		\end{equation}
	\end{itemize}
	\end{proposition}
	\begin{proof}
   At beginning of the algorithm, it holds
   $w_{ij}=0\ (ij\in E)$ and $d_i = \infty\ (i \in N)$. Then ${\cal U} = 2^N$.
   Also all $X$ satisfies (\ref{tilde_g}) by the monotonicity of $f_j\ (j\in M)$.
   
	Next we show: If $X \in {\cal U}$ satisfies (\ref{tilde_g}) at line 4, then $X$ satisfies (\ref{tilde_g}) in the For loop (line 4--8).
	Consider $X\in \mathcal U$ satisfying (\ref{tilde_g}) at line 4 
	with the turn of buyer $i \in N$.
    Suppose 
	that the buyer $i$ clinches 
	$\xi\in \mathbb R^{E_i}_+$ of goods, 
	and that $w,d$ are updated to $w',d'$ in lines 6--8.
	If $i \in X$, then
	for each $F\supseteq E_X$,
	$f(F)-w(F)$ decreases by $\xi(E_i)$.
	Thus the equality (\ref{tilde_g})
	is keeping after lines 6--8.
	If $i\in N\setminus X$, then,  
	by  (\ref{tilde_g}) for $w,d$ and Lemma~\ref{lem:change}~(ii), we have
	\begin{equation*}
	\tilde{g}_{w',d'}(X) \leq f_{w'}(E_X) \leq f(E_X) - w'(E_X) = f(E_X)-w(E_X) = f_w(E_X) = \tilde{g}_{w,d}(X) = \tilde{g}_{w',d'}(X).
	\end{equation*}
	Thus we obtain (\ref{tilde_g}) for new $w',d'$.

	Next we show that no new sets
	are added to ${\cal U}$ in the For loop.
    (In fact, one can show that ${\cal U}$ does not change in the For loop.)
	Take $X \not \in {\cal U}$ at line 4.
	Then there is a minimal minimizer 
	$Z \subset X$ of $\tilde g_{w,d}(X)$ such that
	\begin{align}\label{eqn:bothside}
	f_w(E_X) \geq f_w(E_Z) + d(X \setminus Z) = \tilde g_{w,d}(X).
	\end{align}
	Consider the clinch of buyer $i$ as above.
	If $X$ contains $i$, 
	then both sides of (\ref{eqn:bothside}) decrease by $\xi(E_i)$, and  
	$X$ is not a minimal minimizer for $\tilde g_{w',d'}(X)$.
	If $X$ does not contain $i$, then by (\ref{tilde_g}) for $w,d$ and Lemma~\ref{lem:change}~(ii), 
	we have
	$\tilde g_{w,d}(X) = f_w(Z) + d(X \setminus Z) 
	= f(E_Z) - w(E_Z) + d(X \setminus Z) = f(E_Z) - w'(E_Z) + d(X \setminus Z) \geq f_{w'}(E_Z) + d(X \setminus Z)  \geq \tilde g_{w',d'}(X) = \tilde g_{w,d}(X)$, and
	this means that $Z \subset X$ is a minimizer of $\tilde g_{w',d'}(X)$. 
	Hence $X \not \in {\cal U}$ holds after the clinch.
	
	Finally we show that no sets are added 
	to $\mathcal U$ by the calculation of $d_l$ in line 13.
	For each $X\notin \mathcal U$,
	there exists a proper subset $Z\subset X$ 
	such that the inequality in (\ref{eqn:bothside}) holds.
	Since $f_w$ does not depend on $d$, 
	and $d$ is nonincreasing, 
	this inequality still holds after the recalculation of $d_l$.
	Therefore it still holds that $X\notin \mathcal U$.
	\end{proof}
	Finally we prove Lemma~\ref{lem:change}~(i), 
	which completes the proof of Theorem~\ref{f=g}~(ii) by (\ref{eqn:immediate}).
	\begin{proof}[Proof of Lemma~\ref{lem:change}(i)]
		Consider $X \subseteq N$ with $i \in X$. 
		Notice that $d_i$ is changed to $d'_i = d_i - \xi(E_i)$ after the clinch of $i$.
		Let $Z \in {\cal U}$ be a minimal minimizer of $\tilde g_{w,d}(X)$.
		By Proposition~\ref{min}~(ii), 
		we have $\tilde g_{w,d}(X) = f_w(E_Z) + d(X \setminus Z) = f(E_Z) - w(E_Z) + d(X \setminus Z)$.
		After the clinch, the last quantity decreases by $\xi(E_i)$ 
		regardless whether $Z$ contains $i$ or not.
		This means that $\tilde g_{w',d'}(X) \leq \tilde g_{w,d}(X) - \xi(E_i)$.
		
		We show the converse. 
		Let $Z' \in {\cal U}$ be a minimal minimizer of $\tilde g_{w',d'}(X)$.
		By the monotonicity of ${\cal U}$ 
		(Proposition~\ref{min}~(i)),  $Z'$ belongs to ${\cal U}$ before the clinch.
		Therefore it holds again $\tilde g_{w,d}(X) \leq f_w(E_{Z'}) + d(X \setminus Z') 
		= f(E_{Z'}) - w(E_{Z'}) + d(X \setminus Z')$.
		After the clinch, the last quantity decreases by $\xi(E_i)$, 
		which equals $\tilde g_{w',d'}(X)$.
		This means that  $\tilde g_{w,d}(X) \leq \tilde g_{w',d'}(X) + \xi(E_i)$.
	\end{proof}

	\subsection{Proof of Theorem \ref{pareto}}\label{subsec:pareto}

	Here we prove the pareto optimality (Theorem~\ref{pareto}) in 
	the setting of concave budget constraints (Remark~\ref{rem:concave}).
	In the proof, it is important to analyze how buyers drop out of the auction, or situations when their demands become zero.
	Thanks to Theorem \ref{relation},  
	we can use properties obtained in 
	Goel et al. \cite{GMP2014} for one-sided markets.
	We utilize the notion of the dropping price 
	introduced by them. 
	The {\em dropping price} of buyer $i$, denoted by $\theta_i$, 
	is defined as the first price $c_i$ for which 
	the demand $d_i$ is zero.
	As shown in \cite{GMP2014}, 
	there are three cases when buyer $i$ drops out of the auction.
	\begin{description}
	\item[{\rm Case 1:}] Buyer $i$ clinches his entire demand $\xi(E_i)=d_i$.
	In this case, $\theta_i \leq v_i - \epsilon$.
	After the clinch, it holds $p_i=\phi_i(w(E_i))$.  
	By the concavity of $\phi_i$, 
	the demand never becomes positive.
   	\item[{\rm Case 2:}] Buyer $i$ does not clinch his entire 
	demand but the price reaches his bid, i.e., $\theta_i = v_i'$.
	\item[{\rm Case 3:}] Buyer $i$ does not clinch his entire 
	demand but $p_i = \beta_i  w(E_i)$ and $\theta_i > \beta_i$.
	\end{description}
	We here call the event of case 2 or 3 the {\em unsaturated drop} (of buyer $i$).
	Recall that $\beta_i := \lim_{x \to +0}\phi_i(x)/x$ 
	represents the angle of $\phi_i$ at $0$.
	In the usual budget constraints, 
	case 1 means $p_i = B_i$ and 
	case 3  never occurs (by $\beta_i = \infty$).

	We use the following intriguing property 
	to prove the pareto optimality (Theorem \ref{pareto}).
	Here a subset $X\subseteq N$ is said to be 
	{\it tight} with respect to transaction $w$ if 
	$w(E_X)=f(E_X)(=g(X))$.
　	\begin{proposition}[Goel et al. \cite{GMP2014}]
	\label{dropping} 
	 Let $i_1,i_2,\ldots,i_t$ be the buyers doing unsaturated drop, where they are sorted in reverse order of their drops; 
	 then $\theta_{i_1}\geq\cdots\geq\theta_{i_t}$.
	 For each $k =1,2,\ldots,t$, 
	 let $X_k$ denote the set of buyers having
	 a positive demand just before the drop of $i_k$. 
	 Then we have the following:
	 \begin{itemize}
	 	\item[{\rm (i)}] $\emptyset= X_0\subset X_1
	 	\subset  X_2\subset 
	 	\cdots\subset  X_t=N$ is a chain of tight sets.
	 	\item[{\rm (ii)}] For $k=1,2,\ldots,t$, it holds $i_k\in X_k\setminus X_{k-1}$.
	 	 \item[{\rm (iii)}] For $i\in X_k\setminus (X_{k-1}+i_k)$, it holds
	 	 $\theta_i\in \{\theta_{i_k}-\varepsilon,\theta_{i_k}\}$.
	 \end{itemize}
	\end{proposition}
	
   The proof of PO in our mechanisms is a modification of that of \cite[Theorem 4.6]{GMP2014}.

	\begin{proof}[Proof of Theorem \ref{pareto}]
	We denote the set of nonvirtual buyers 
	by $N^{\ast}$.
	By (ICb) (Theorem~\ref{goel} or Corollary \ref{ICb}), 
	we may assume that each buyer $i$ reports 
	true valuation, i.e., $v'_i=v_i\ (i\in N^{\ast})$.
	Let $\mathcal A=(w,p,r)$ be the allocation obtained by 
	Algorithm \ref{algo2} 
	(including the allocation to the virtual buyers).
	Consider an arbitrary allocation 
	(without virtual buyers) 
	$\mathcal A'=(w',p',r')$ 
     such that 
	\begin{align}
	\label{assumption0}
	\sum_{i\in N^{\ast}}
	p'_i & \geq\sum_{j\in M} r'_j \\
	\label{assumption1}
	v_i w(E_i)-p_i&\leq v_i w'(E_i)-p'_i\quad(i\in N^{\ast}),\\
	\label{assumption2}
	(r_j-p_{n+j})+\rho_j w_{(n+j)j} &\leq 
	r'_j+\rho_j(f_j(E_j)-w'(E_j))\quad(j\in M).
	\end{align}
	Our goal is to show that all inequalities 
	in (\ref{assumption1}) and (\ref{assumption2}) hold in equality.
	Here $\mathcal A'$ is extended to an allocation 
	with virtual buyers by:
	\begin{align*}
	w'_{(n+j)j}&:=f_j(E_j)-w'(E_j)\quad (j\in M), \\
	p_{n+j}&:=0 \quad(j\in M).
	\end{align*}
	In the sequel, virtual buyers are included in the market, i.e., $(n+j)j \in E_j$. 
	Then we have
	\begin{align}
	\label{w'=f}
	w'(E)=\sum_{j\in M}f_j(E_j)=f(E).
	\end{align}
	Also (\ref{assumption1}) and (\ref{assumption2}) are 
	rewritten as
	\begin{align}
	\label{assumption1'}
p_i-p'_i & \geq 	v_i (w(E_i)-w'(E_i)) \quad (i\in N^{\ast}), \\
	\label{assumption2'}
	\rho_j(w'_{(n+j)j}-w_{(n+j)j}) & \geq 
	(r_j-r'_j)-(p_{n+j}-p'_{n+j})
	\quad (j\in M).
	\end{align}
 We show
	\begin{claim}[informal]
		By adding equalities and inequalities that include all in (\ref{assumption1'}) and (\ref{assumption2'}), one can deduce 
		\begin{equation}
		\label{budgetbalance}
		\sum_{i\in  N^{\ast}}(p_i-p'_i) 
		\geq \sum_{j \in M}
		((r_j-r'_j)-(p_{n+j}-p'_{n+j})).
		\end{equation}
	\end{claim}
	We first complete the proof assuming this claim.
		By substituting (SBB) $\sum_{i\in N^{\ast}}p_i
		=\sum_{j\in M}(r_j-p_{n+j})$ for
		(\ref{budgetbalance}), we obtain
		$
		\sum_{i\in N^{\ast}}p'_i \leq 
		\sum_{j\in M}(r'_j-p'_{n+j})$.
		Since $p'_{n+j}=0\ (j\in M)$, we have
		\begin{equation}
		\label{revenue-payment}
		\sum_{i\in N^{\ast}}p'_i\leq\sum_{j\in M}r'_j.
		\end{equation}
		By (\ref{assumption0}), the equality 
		holds in (\ref{revenue-payment}).
	   Consequently, all inequalities in  
		(\ref{assumption1'}) and (\ref{assumption2'})
		must hold in equality, which implies the pareto optimality (Theorem~\ref{pareto}).

     Finally we prove the claim.	
		Consider the buyers $i_1,i_2,...,i_t$ and the chain 
		$\emptyset= X_0\subset X_1\subset  X_2\subset
		\cdots\subset  X_t=N$ of tight sets in Proposition~\ref{dropping}.
		Then $i_k\in X_k\setminus X_{k-1}$, and 
		\begin{align}
		f(E_{ X_k})&=w(E_{ X_k})\quad(k\in \{0,1,2,\ldots,t\}),\\
		f(E_{ X_t})&=f(E)=w'(E)=w'(E_{ X_t}), \label{eqn:f(E)=w'(E)}
		\end{align}
		where (\ref{eqn:f(E)=w'(E)}) follows from (\ref{w'=f}).
	For each $k\in \{1,2,\ldots,t\}$, 
	we define the following sets: 
	\begin{align*}
	 X^{\ast}_k:&= X_k\cap N^{\ast},\\
	 Z_k:&= X_k\setminus  X_{k-1},\\ 
	 Z^{\ast}_k:&= Z_k\cap N^{\ast}. 
	\end{align*}
	
	We show: For each $k\in\{1,2,\ldots,t\}$ and 
	$i\in Z^{\ast}_k$, it holds 
	\begin{equation}
	\label{assumption1''}
	p_i-p'_i\geq \theta_{i_k}(w(E_i)-w'(E_i)).
	\end{equation}
The proof is a case-by-case analysis; 
recall the above three cases of the drop of buyer $i$.

		Case 1-1: $i\neq i_k$ and $w(E_i)\geq w'(E_i)$.

		In this case, $i$ clinched his entire demand at his dropping price $\theta_i$. 
		Hence it holds $\theta_i\leq v_i-\varepsilon$.
		By $\theta_i\in \{\theta_{i_k}-\varepsilon,\theta_{i_k}\}$ (Proposition~\ref{dropping}~(iii)),
		we have $\theta_{i_k}\leq v_i$.
		By (\ref{assumption1'}), we obtain
		$p_i-p'_i\geq v_i(w(E_i)-w'(E_i))\geq 
		\theta_{i_k}(w(E_i)-w'(E_i)).$

		Case 1-2: $i\neq i_k$ and $w(E_i)< w'(E_i)$. 
		
		Also in this case, $i$ clinched his entire 
		demand at price~$\theta_i$.
		By the definition 
		of demand, for any 
		$\kappa>0$, it holds $p_i+\theta_i\kappa>\phi_i(w(E_i)+\kappa)$. 
		By substituting $\kappa=w'(E_i)-w(E_i) > 0$, 
		we obtain 
		$p_i+\theta_i (w'(E_i)-w(E_i))>
		\phi_i(w'(E_i))\geq p'_i$.
		By $\theta_i\leq \theta_{i_k}$ (Proposition~\ref{dropping}~(iii)),
		it holds 
		$p_i-p'_i > \theta_i(w(E_i)-w'(E_i)) \geq 
		\theta_{i_k}(w(E_i)-w'(E_i))$; notice that $w(E_i)-w'(E_i)$ is negative.
		
		Case 2: $i=i_k$ and $\theta_{i}=v_{i}$.
		
		By (\ref{assumption1'}), it holds
		$p_{i}-p'_{i}\geq v_{i}(w(E_{i})-w'(E_{i}))=
		\theta_{i}(w(E_i)-w'(E_i))$.
		
		Case 3-1: $i=i_k$, $\theta_{i}>\beta_{i}$, 
		$p_{i}=\beta_{i}w(E_{i})$, and $w(E_{i}) \geq w'(E_{i})$. 
		
		By $\theta_{i}\leq v_{i}$ and 
		(\ref{assumption1'}), 
		it holds $p_{i}-p'_{i}\geq v_{i}(w(E_{i})-w'(E_{i})) \geq
		\theta_{i}(w(E_i)-w'(E_i))$.
		
			Case 3-2: $i=i_k$, $\theta_{i}>\beta_{i}$, 
		$p_{i}=\beta_{i}w(E_{i})$, and $w(E_{i}) < w'(E_{i})$.
		
		Since $p'_{i}\leq 
		\beta_{i} w'(E_{i})$,
		it holds $p_{i}-p'_{i}\geq \beta_{i}
			(w(E_{i})-w'(E_{i}))>\theta_{i}(w(E_i)-w'(E_i))$.
	
	In this proof of (\ref{assumption1''}), 
	we use (\ref{assumption1'}) 
	for cases 1-1, 2, and 3-1.
	For other cases (1-2 and 3-2), 
	the strict inequality holds in (\ref{assumption1''}).
	It will turn out that those cases never occur.
	Namely (\ref{assumption1'}) is 
	used to deduce (\ref{assumption1''})
	for each $i \in N^*$.

	Next we prove: for each $k\in \{1,2,\ldots,t\}$, it holds
	\begin{align}
	\label{induction}
	\sum_{i\in  X^{\ast}_k}(p_i-p'_i) 
	&\geq \sum_{j \in M:\,n+j\in  X_k}
	((r_j-r'_j)-(p_{n+j}-p'_{n+j})) +\theta_{i_k}
	(w(E_{ X_k})-w'(E_{ X_k}))
	\end{align}
	The proof uses induction on $k$.
	In the case of $k=1$,
	\begin{align*}
	\sum_{i\in  X^{\ast}_1}(p_i-p'_i)&=\sum_{i\in  Z^{\ast}_1}
	(p_i-p'_i)\geq\sum_{i\in  Z^{\ast}_1}
	\theta_{i_{1}}(w(E_i)-w'(E_i)) =\sum_{i\in  X^{\ast}_1}\theta_{i_{1}}(w(E_i)-w'(E_i)) \\
	&= \sum_{i\in  X_1}\theta_{i_{1}}(w(E_i)-w'(E_i)) + 
	\sum_{j: n+j \in X_1} \theta_{i_{1}}(w'_{(n+j)j}-w_{(n+j)j}) \\
	& = \theta_{i_{1}}(w(E_{X_1})-w'(E_{X_1}))+ \sum_{j: n+j \in X_1} \rho_{j}(w'_{(n+j)j}-w_{(n+j)j}),
	\end{align*}
	where we use (\ref{assumption1''}) with all $i \in Z_1^*$  for the first inequality and the fourth equality follows from 
	$\rho_j= v'_{n+j} = \theta_{i_1}$ since 
	each virtual buyer never
	clinch his entire demand and $\beta_{n+j} = \infty$. 
	Now  by substituting (\ref{assumption2'}) 
	for the summation, we obtain  (\ref{induction}) for~$k=1$.
	Suppose that 
	(\ref{induction}) holds in $k\geq 1$. 
	By inductive assumption, we have 
	\begin{align*}
	\sum_{i\in  X^{\ast}_{k+1}}(p_i-p'_i)
	& =\sum_{i\in  X^{\ast}_{k}}(p_i-p'_i)
	+\sum_{i\in  Z^{\ast}_{k+1}}(p_i-p'_i)\\
	& \geq \sum_{j \in M:\,n+j\in  X_k}
	((r_j-r'_j)-(p_{n+j}-p'_{n+j}))  \\ 
	& \quad + \theta_{i_k}
	(w(E_{ X_k})-w'(E_{ X_k})) + \theta_{i_{k+1}} (w(E_{Z_{k+1}^*}) - w'(E_{Z_{k+1}^*})),
	\end{align*}
	where we use (\ref{assumption1''}) with all $i \in Z_{k+1}^*$.
	Since $\theta_{i_k} \geq \theta_{i_{k+1}}$ and $w(E_{X_k}) = f(X_k)  \geq w'(E_{X_k})$, the second and third terms are further calculated as
	\begin{align*}
	& \theta_{i_k}
	(w(E_{ X_k})-w'(E_{ X_k})) + \theta_{i_{k+1}} (w(E_{Z_{k+1}^*}) - w'(E_{Z_{k+1}^*})) \\
	& \geq \theta_{i_{k+1}}  (w(E_{ X_k})-w'(E_{ X_k}) + w(E_{Z_{k+1}^*}) - w'(E_{Z_{k+1}^*})) \\
	& = \theta_{i_{k+1}}( w(E_{X_{k+1}}) -  w'(E_{X_{k+1}}) ) + \sum_{j: n+j \in Z_{k+1}} \theta_{i_{k+1}} (w'_{(n+j)j} - w_{(n+j)j}) \\
	& \geq \theta_{i_{k+1}}( w(E_{X_{k+1}}) -  w'(E_{X_{k+1}})) +  \sum_{j:\,n+j\in  Z_{k+1}}
	((r_j-r'_j)-(p_{n+j}-p'_{n+j})),  
	\end{align*}
	where we use $\theta_{i_{k+1}} = \rho_j$ again and use (\ref{assumption2'}) for 
	$j \in M$ with $n+j \in Z_{k+1}$.
	Gathering the above two inequalities, we obtain (\ref{induction}) for $k+1$, and complete the proof of (\ref{induction}).

    Now consider the case of $k = t$ in (\ref{induction}).
	By $X_t = N$ and (\ref{eqn:f(E)=w'(E)}), 
	the second term of the right hand side vanishes, and we obtain (\ref{revenue-payment}).
	Here we used all inequalities in (\ref{assumption2'}) and (\ref{assumption1''}) to deduce (\ref{induction}) for $k=t$.
	Therefore they must hold in equality. 
	In particular, cases 1-2 and 3-2, 
	which derive strict inequality in (\ref{assumption1''}), never occur.
	This means that we also used all inequalities in (\ref{assumption1'}) 
	to deduce (\ref{assumption1''}). 
	Thus (\ref{assumption1'}) must hold in equality.
	This completes the proof of Theorem~\ref{pareto}.
	\end{proof}
	
	\section*{Acknowledgements.}
	This work was partially supported by JSPS KAKENHI 
	Grant Numbers JP25280004, JP26330023, JP26280004, JP17K00029,
	and by JST ERATO Grant Number JPMJER1201, Japan.
	
\bibliographystyle{abbrv}
\bibliography{polyhedral}
\end{document}